\documentclass[a4paper,11pt]{article}

\usepackage{amssymb,amscd,amsfonts,amsmath,graphicx,amsthm}
\usepackage{latexsym}
\usepackage{pgf,tikz}
\usepackage[ansinew]{inputenc}%para usar acentos comuns
\usepackage{mathrsfs}
\usepackage{graphicx}
\usepackage{colortbl}
\usepackage{rotating}
\usepackage[usenames,dvipsnames]{pstricks}
\usepackage{calrsfs}
\usepackage{pdfpages}
\usepackage[all]{xy}
\usepackage{epsfig}
\usepackage{pst-grad} % For gradients
\usepackage{pst-plot} % For axes
\usepackage{fancyhdr}
\usepackage[top=3cm,bottom=2cm,left=3cm,right=2cm]{geometry} %%margens
\usepackage{imakeidx}
\usepackage{enumerate}
\makeindex %%%%%indice remisivo%%%%%%
\usepackage{pstricks,pst-node,pst-plot,pst-tree}
\usepackage{pstricks-add}
\usepackage{pgf,tikz}
\usetikzlibrary{arrows}
\newcommand{\Q}{{\mathbb{Q}}}

\newcommand{\C}{{\mathbb{C}}}

\newtheorem{Theo}{Theorem}%[section]

\newtheorem{cor}[Theo]{Corollary}
\newtheorem{lema}[Theo]{Lemma}
\newtheorem{defi}[Theo]{Definition}

\newtheorem{Exam}[Theo]{Example}

\newtheorem{Rema}[Theo]{Remark}

\newcommand{\G}{\mathcal{G}}
\newcommand{\F}{\mathcal{F}}
\newcommand{\Msup}{\mathbb{M}}
\newcommand{\Tri}{\mathbb{T}}
\newcommand{\Pol}{\mathbb{P}}

\newcommand{\LL}{\mathcal{L}}
\newcommand{\A}{\mathbb{A}}
\newcommand{\B}{\mathbb{B}}

\begin{document}
	\title{\bf New topological subsystem codes from semi-regular tessellations}
	\author{Eduardo Brandani da Silva$^1$ and Evandro Mazetto Brizola$^2$ \\
	{\it Maringa State University - Department of Mathematics} \\
{\it $^1$ebsilva@uem.br} \\
{\it $^2$evandro.brizola@hotmail.com}}

	\date{}
	
	\maketitle
	
	\noindent
	\renewcommand{\thefootnote}{\fnsymbol{footnote}}
	%\footnotetext{2010 \emph{Mathematics Subject Classification}: {47B10, 47B02, 47B37.}}
	\footnotetext{\emph{Key words and phrases}: Quantum error-correcting codes, subsystem codes, topological subsystem codes, semi-regular tessellations}

	\begin{abstract}
In this work, we present new constructions for topological subsystem codes using semi-regular Euclidean and hyperbolic tessellations. They give us new families of codes, and we also provide a new family of codes obtained through an already existing construction, due to Sarvepalli and Brown. We also prove new results that allow us to obtain the parameters of these new codes.
	\end{abstract}
	
	%\pagebreak
	
	\medskip
	
	%\tableofcontents
	
\section{Introduction}

The first quantum computing model described in 1982 by Benioff \cite{benioff}, although based on quantum kinematics and dynamics, was, in the computational sense established in \cite{deutsch}, effectively classical. Also in 1982, Feynman \cite{feynman} proposed the first computer model based on the principles of quantum mechanics. This model was the closest thing to a universal quantum simulator. As early as 1985, a major breakthrough occurred in this area when Deutsch \cite{deutsch} described an entirely quantum model of a quantum version of the Turing machine.

One of the biggest difficulties in performing quantum computing is quantum decoherence, which is the issue of quantum decay, which was warned by Unruh \cite{unruh}. One way to get around this effect is to use quantum error-correcting codes, which are mathematical tools used to encode the physical states of a quantum system.

The first quantum error-correcting codes were inspired by the classical error-correcting codes introduced in 1948 by Shannon \cite{shannon}. In 1995, Shor \cite{shor} introduced the first quantum error-correcting code. This code was a combination of two other three-qubit codes, where one of them protects the quantum state against bit flip errors and the other protects the quantum state against phase-flip errors, where only the first of them has a classical analog.

One of the biggest advances in quantum coding occurred in 1996 with the construction of the CSS quantum codes by Calderbank and Shor \cite{calderbank} and Steane \cite{steane}. From these codes, one of the most important classes of quantum codes known to date was generated, the stabilizer quantum codes \cite{codigoestabilizador}.

The following year, Kitaev \cite{kitaev} developed a new code called Kitaev's toric code, which is part of a class of codes known as topological codes. A great advantage and novelty of this code was the fact that its stabilizer generators acted on a small number of qubits in its neighborhood. Codes with this property are said to be geometrically local. An advantage of stabilizing generators acting on a small number of qubits in their neighborhood is that this makes quantum words resistant to local noise.

Another important group of quantum codes that is also part of the stabilizer codes are the subsystem codes \cite{salah}, \cite{baconsub}, \cite{poulin}. These codes are the result of applying the stabilizer formalism to quantum error correction of the operator \cite{9poulin}, \cite{8poulin}, which generalizes standard quantum error correction theory and provides a unified framework for active correction of errors and passive error prevention techniques, making subsystem codes one of the most versatile class of error-correcting quantum codes. 

Over the years, numerous works have been carried out in this area. A great step was given by the construction of topological subsystem codes done by Bombín \cite{bombinsub}, which made use of trivalent and $3$-colorable tessellations. Although in \cite{baconsub} tessellations have already been used to construct subsystem codes, due to the stabilizer generators of these codes being non-local, they cannot be considered topological.

Suchara, Bravyi, and Terhal \cite{terhal} established some conditions for determining topological subsystem codes and proved that topological subsystem codes can be seen as a kind of generalization of Kitaev's honeycomb model, for trivalent hypergraphs. Also in \cite{terhal}, a necessary and sufficient condition was presented about measuring the error syndrome of a subsystem code, which is a very important result and is used to guarantee that a subsystem code is topological.

Sarvepalli and Brown \cite{sarvepalli} presented a new construction of topological subsystem codes inspired by \cite{terhal}, where they presented some new code families. This is one of the few works where code parameters are provided, due to the difficulty of determining certain standards within this class of codes.

The work is divided into the following sequence: In Section \ref{cap5}, we revise the subsystem codes and present their main elements, such as the groups necessary for their construction and definition. In Section \ref{codigosubsistemtopolo}, we will address topological subsystem codes, where we will discuss some important advantages of these codes in relation to topological quantum codes. We will also deal with the main aspects of the construction of such codes done in \cite{bombinsub} and their generalization done in \cite{terhal} and \cite{sarvepalli}. Finally, we will look at two families of topological subsystem codes presented in \cite{sarvepalli}. In Section \ref{construcaomeussubcode}, we construct a new hypergraph $\Gamma_h$ and prove that this hypergraph allows the construction of subsystem codes. We also present, in general, some stabilizers and other relevant data about this construction. In Section \ref{familiassubsistem}, we present four families of topological subsystem codes, where two of these families come from particular cases of the hypergraph $\Gamma_h$ constructed in Section \ref{construcaomeussubcode}. The third family comes from the construction given in \cite{sarvepalli}, but from a case that was not addressed by the authors, and the fourth family comes from a hypergraph built from tessellation $\{p,3,4,3\}$ in a similar way to that used for previous families. Also in this section, we prove that these code families are topological subsystem codes and provide some tables with parameters.

\section{Subsystem codes}\label{cap5}

In this section, we will cover subsystem codes, which are also stabilizer codes. According to Bombín \cite{bombinsub}, these codes are the result of applying the stabilizer formalism to the quantum error correction of the operator \cite{9poulin}, \cite{8poulin}. As can be seen in \cite{poulin}, Poulin claims that operator quantum error correction generalizes the standard quantum error correction theory and provides a unified framework for active error correction and passive error prevention techniques, which, according to Aly and Klappenecker \cite{salah}, made subsystem codes the most versatile class of error-correcting quantum codes known up to that time.

In subsystem codes, not all logical qubits that form the $\textbf{C}$ code space are used as logical qubits, that is, they are used to encode information. A portion of these qubits are considered gauge qubits, which do not encode any information. In addition to other consequences, these qubits absorb the effects of errors that occur in them. 

The definition of subsystem code is similar to stabilizer codes, as we also have a stabilizer group $\textbf{S}$, and the code space $\textbf{C}$ is also given in a similar way.

We have a code space $\textbf{C}$ where information is encoded. From this code space, which is a subspace of a Hilbert space $H$, we can decompose $H$ into the direct sum $H=\textbf{C}\oplus \textbf{C}^\perp$.

In subsystem codes, the code space $\textbf{C}$ is no longer fully used. It admits a subsystem structure, so we can write $\textbf{C}$ as $\textbf{C}=\A\otimes \B$, where $\A$ is the subsystem of $\textbf{C }$ that is used to encode the information. The subsystem $\B$ does not encode information, and errors that occur in it are ignored. Therefore, we say that the subsystem $\B$ contributes with gauge degrees of freedom. Because we use subsystems, the code is called subsystem code. To obtain the subsystems $\A$ and $\B$, we need two groups, which are fundamental for the construction of the subsystem codes.

The first of these is the gauge group, denoted by $\G$, which is a subgroup of the Pauli group $P_n$. More precisely, this group will be a normal subgroup of $N(\textbf{S})$ and will be generated by the stabilizer group $\textbf{S}$, by $\langle i \rangle$, and by an arbitrary subset of $X_j'$ and $Z_j'$ with $j>s$, where $s$ is the number of independent stabilizer generators.

The second fundamental group for the construction of subsystem codes is the logical group $\LL=N(\textbf{S})/\G$, which is a group because $\G$ is a normal subgroup of $N(\textbf{S})$.

The elements of $\G$ are required to commute with the elements of $\LL$. According to Poulin \cite{poulin}, a consequence of their commuting is that the generators $X_j'$ and $Z_j'$ of $\G$ must always appear in pairs. Therefore, we can assume, without loss of generality, that
\begin{equation}
\G=\langle i,S_1,\ldots,S_s,X_{s+1}',\ldots,X_{s+r}',Z_{s+1}',\ldots,Z_{s+r}'\rangle,
\end{equation}
with $s+r\leq n$. We can ignored the phase factor, then the gauge group to be considered is only
\begin{equation}
\G=\langle \textbf{S},X_{s+1}',\ldots,X_{s+r}',Z_{s+1}',\ldots,Z_{s+r}'\rangle.
\end{equation} 
Therefore, it follows that
\begin{equation}
\LL \simeq \langle X_{s+r+1}', Z_{s+r+1}',\ldots,X_{n}',Z_{n}'\rangle.
\end{equation}

Since the elements of $\G$ commute with the elements of $\LL$ and $\G\times\LL\simeq N(\textbf{S})$, it can be seen from \cite{22poulin} that it is possible to induce a subsystem structure in code space $\textbf{C}$, so that we have $\textbf{C}=\A\otimes \B$, where the group $\G$ acts trivially on $\A$ and $\LL$ acts trivially on $\B$. It also follows that $\A\simeq (\C^2)^{\otimes k}$ and $\B\simeq (\C^2)^{\otimes r}$. See that if we use $\B$ as a one-dimensional space, then we will have exactly the usual quantum error-correcting codes. Note that $\LL$ acts as the logical operators of the stabilizer codes. Thus, the biggest novelty so far in the subsystem codes is the presence of the gauge group $\G$.

The action of the gauge operators in the code space $\textbf{C}$ induces the following equivalence relation in $\textbf{C}$. Let $|\psi\rangle , |\psi\rangle'\in \textbf{C}$ be, we have
\[
|\psi\rangle \sim |\psi\rangle' \Leftrightarrow\exists g\in \G\ \ \textrm{tal que}\ \ |\psi\rangle' = g|\psi\rangle \,.
\]

By definition, equivalent states $|\psi_{\A}\rangle \otimes |\psi_{\B}\rangle$ and $|\psi_{\A}\rangle\otimes |\psi'_{\B}\ rangle$ carry the same information, even if $|\psi_{\B}\rangle$ and $|\psi'_{\B}\rangle$ are different.

Another important fact about the gauge group is that the stabilizer group $\textbf{S}$ of the subsystem code can be described using $\G$ due to the identity $\textbf{S}=\G\cap C(\G) $, where $C(\G)$ is the centralizer of $\G$ in the Pauli group. Thus, to characterize the subsystem codes, we only need to have their gauge group $\G$.

It is worth noting that in the stabilizer codes we had $n=k+s$. In the subsystem codes, we divide the $n$ virtual qubits into three sets: $s$ stabilizer qubits, $r$ gauge qubits, and $k$ logical qubits, where the physical qubits are now given by $n=k+r+s $. The operators related to the first set are given by the independent stabilizer generators $S_1,\ldots, S_s$ that fix the $2^{n-s = r+k}$-dimensional code space $\textbf{C}$. The second is generated by $X_{s+j}'$ and $Z_{s+j}'$ with $j=1,\ldots,r$, which act on the $r$ virtual qubits of subsystem $\B $, which do not encode useful information and their only purpose is to absorb gauge transformations. Therefore, the gauge group $\G$ has dimension $2r+s$. Finally, the third set is given by the remaining operators, which are logical operators as in stabilizer codes. Thus, we denote these operators by $\overline{X}_j$ and $\overline{Z}_j$ with $j=1,\ldots,k$, which generate the set $\LL$ and act only on $k$ logical qubits of subsystem $\A$.

Since the elements of $\LL$ are the logical operators of the subsystem code, they are elements of the normalizer of $\textbf{S}$, that is, they are elements of
\begin{equation}
N(\textbf{S}) = \langle\G, \overline{X}_1,\overline{Z}_1,\ldots,\overline{X}_k,\overline{Z}_k\rangle.
\end{equation}

Two other operators that appear frequently in the literature on subsystem codes are bare logical operators and dressed logical operators. Bare logical operators refer to the elements of $C(\G)\setminus\G$. These elements preserve the code space $\textbf{C}$ and act trivially on the gauge qubits. The dressed logical operators refer to the elements of $N(\textbf{S})\setminus\G$. Since we have the identity $N(\textbf{S})=C(\G)\cdot\G$ hold, where $``\cdot "$ is the product of the elements of $C(\G)$ by the elements of $ \G$, we have that any dressed logical operator can be written as a product of a bare logical operator and a gauge operator.

As can be seen in \cite{bravyi} the distance of subsystem codes is defined as follows.
\begin{defi}
The distance $d$ of a subsystem code is the minimum weight of a Pauli operator that commutes with all stabilizers and acts non-trivially on the logical subsystem $\A$, that is, it will be the minimum weight of a dressed logical non-trivial operator, that is,
\begin{equation}
d = \min_{P\in N(\textbf{S})\setminus \G} |P| = \min_{\substack{P\in C(\G)\setminus\G \\ G\in\G}} |PG|,
\end{equation}
where $|P|$ and $|PG|$ are the weights of the Pauli operators $P$ and $PG$, respectively.
\end{defi}

Thus, we have four parameters for the subsystem codes, which are denoted by $\left[\left[n,k,r,d\right]\right]$, where $n$ is the number of physical qubits, $k$ is the number of encoded qubits, $r$ is the number of gauge qubits, and $d$ is the distance from the code.

For more information about the subsystem codes, their construction, the sets presented here, and to see about error identification and correction, we recommend \cite{salah}, \cite{bravyi}, \cite{breuck}, and \cite{ poulin}.

%Embora esse código não seja chamado na literatura de código de subsistema topológico e, sim, somente de código de subsistema, pois seus geradores, como veremos, são não-locais, resolvemos colocá-lo nessa seção pelo mesmo fazer uso de tesselação na sua construção.

\section{Topological subsystem codes}\label{codigosubsistemtopolo}

In this section, we present the topological subsystem codes. These codes were first developed by Bombín \cite{bombinsub} in 2010, who made use of trivalent and $3$-colorable tessellations. In the same year, they were generalized to trivalent hypergraphs satisfying certain conditions by Suchara, Bravyi, and Terhal \cite{terhal}. Subsequently, several works on these codes emerged, also working in the hyperbolic environment. In 2012, Sarvepalli and Brown \cite{sarvepalli} added one more condition about hypergraphs that would be necessary in \cite{terhal}, but had not been considered. They also presented two families of topological subsystem codes and provided the parameters of these codes, which, according to them, is not a trivial task for subsystem codes. These code families built by Sarvepalli and Brown follow the construction model for obtaining trivalent hypergraphs presented in \cite{terhal}. These three works served as inspiration and provided the necessary tools for building our families of topological subsystem codes that we will present in the next section.

Topological subsystem codes have all the characteristic properties of topological quantum codes, such as that we can visualize the code operators, which according to Breuckmann \cite{breuck} makes the work more easy and intuitive, and also the property that the operators correspond to the interaction between neighboring qubits, which means we have a more physical interpretation of the codes. These codes also have some advantages that come from subsystem codes, such as the fact that the codes are local. As we will see, one of the main advantages compared to topological quantum codes is that these codes require only two neighboring qubits at a time to perform the syndrome measurement for error correction. To see other advantages of topological subsystem codes, we recommend \cite{bombinsub}, \cite{breuck}, and \cite{terhal}. The construction of topological subsystem codes combines the following properties:

\begin{itemize}
\item[$(C_1)$] The stabilizer group $\textbf{S}=\G\cap C(\G)$ has spatially local generators, that is, the number of elements different from the identity (their weight) is within a disk of radius $O(1)$. The elements of $\textbf{S}$ are identified with the homologically trivial loops of the tessellation;
\item[$(C_2)$] Syndrome extraction can be done by measuring the eigenvalues of the link operators, which are $2$-local (2-qubits);
\item[$(C_3)$] The code encodes one or more logical qubits. The bare logical operators (in $C(\G)\setminus \G$) can be identified with the homologously non-trivial loops.
\end{itemize}

\begin{defi}\label{defisubtopolo}
A subsystem code that satisfies the properties $(C_1)$, $(C_2)$, and $(C_3)$ is called a topological subsystem code.
\end{defi}

Property $(C_1)$ requires that the stabilizer operators act on a small number of qubits. Property $(C_3)$ tells us that non-trivial undetectable errors increase with the size of the tessellation and are non-local, which according to Breuckmann \cite{breuck} is an advantage, since the chance of occurring non-local errors is much smaller than local errors. As can be seen in \cite{26terhal}, these two properties are often used to describe some properties of stabilizer codes. Finally, property $(C_2)$ is the main advantage of these codes compared to topological quantum codes.

If only the property $(C_3)$ is not satisfied, then we obtain a subsystem code that does not encode logical qubits ($k=0$), that is, it does not encode information. An example that satisfies $(C_1)$ and $(C_2)$ but does not satisfy $(C_3)$ is Kitaev's Honeycomb model, which can be seen at \cite{27terhal}. In \cite{terhal}, it is proved that the subsystem code built into this tessellation does not encode any information. It is also proved that any tessellations or trivalent graph do not encode information, seen as subsystem code.

Now we will address the main aspects of the construction of topological subsystem codes carried out by Bombín in \cite{bombinsub}, by Suchara, Bravyi, and Terhal in \cite{terhal}, and by Sarvepalli and Brown in \cite{sarvepalli}. But before that, we will define some basic concepts about graph theory. To see more about graph theory, we recommend \cite{bollobas}, \cite{diestel}, and \cite{gross}.

\begin{defi} 
A graph is a pair $\Gamma = (V,E)$, where $E$ is a subset of the set $V\times V$ of unordered pairs of $V$, that is, the elements of $V$ are the vertices of the graph $\Gamma$ and the elements of $E$ are the edges of the graph $\Gamma$.
\end{defi}

\begin{defi} 
Let $\Gamma=(V,E)$ and $\gamma = (V',E')$ be two graphs. If $V'\subseteq V$ and $E'\subseteq E$, then $\gamma$ is said to be a subgraph of $\Gamma$.
\end{defi}

\begin{defi} 
A hypergraph is a pair $\Gamma_h = (V,E)$, where the elements of $E$ are subsets of the elements of $V$. If the subsets are all of size two, then $\Gamma_h$ is a standard graph. And, a trivalent hypergraph is a hypergraph $\Gamma_h$ such that every vertex belongs to exactly three distinct edges. It is also said that such a hypergraph has a degree of $3$.
\end{defi}

\begin{defi} 
Let $\Gamma_h = (V,E)$ be a hypergraph. Any element of $E$ whose size is greater than $2$ is called a hyperedge, and its rank is its size. The rank of a hypergraph is the maximum rank of its edges.
\end{defi}

The construction made by Bombín starts from a trivalent and $3$-colorable tessellation $\Gamma$ and takes the dual tessellation $\Gamma^*$, which is a tessellation formed by triangles. Each edge of the dual tessellation is transformed into a four-sided face, and each vertex is transformed into a face with as many sides as the valence of the vertex, which is already done automatically when generating the four-sided faces, as can be seen in Figure \ref{bombin2}-($a$). This new tessellation will be denoted by $\overline{\Gamma}$.

In \cite{bombinsub}, edges are identified as being of three types: solid, dashed, and dotted edges. However, to avoid confusion and to standardize with our construction and the construction and language carried out in \cite{sarvepalli}, we will color the edges of this new tessellation as follows: We color the triangles blue, and the edges of the faces obtained for each vertex alternate between green and red colors, as can be seen in Figure \ref{bombin2}-($b$).

\begin{figure}[!h]
	\centering
	\includegraphics[scale=0.9]{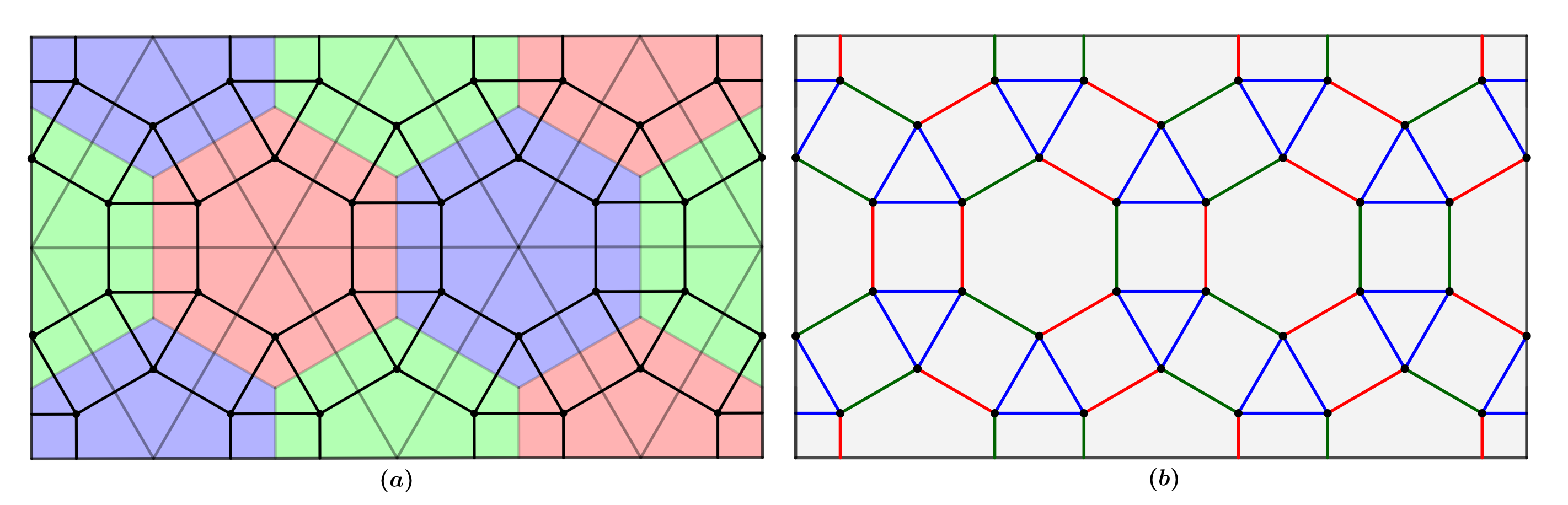} 
	\caption{($a$) Tessellation $\{6,6,6\}$, its dual $\Gamma^*$ and the new tessellation $\overline{\Gamma}$ given by black edges. ($b$) Colored edges of tessellation $\overline{\Gamma}$.} \label{bombin2}
\end{figure}

The operators acting on each edge are called link operators (or edge operators) and are given as follows: Each edge $e'=(u,v)$ is associated with an edge operator $\overline{K}_{e'}\in\{X_uX_v, Y_uY_v, Z_uZ_v\}$, where $\overline{K}_{e'} = Z_uZ_v$ if $e'$ is an edge of the triangle (blue edge), $\overline{K}_{e'} = X_uX_v$ if $e'$ is a red edge, and $\overline{ K}_{e'}=Y_uY_v$ if $e'$ is a green edge. Note that the edge operators are all Hermitian.

Using the edge operators, the gauge group $\G$ is defined as
\begin{equation}\label{grupogauge}
\G:=\langle\overline{K}_{e'};\ e'\in \overline{E}\rangle,
\end{equation}
where $\overline{K}_{e'}$ are the edge operators and $\overline{E}$ is the set of edges of the tessellation $\overline{\Gamma}$. Note that the gauge group generators $\G$ are $2$-local.

As we saw in Section \ref{codigosubsistem}, subsystem codes can be characterized through the gauge group $\G$ due to the identity $\textbf{S} = \G\cap C(\G)$.

Using string operators, in \cite{bombinsub} is determined $C(\G)$. The nature of strings differs from the case of topological quantum codes. 

Given any subgraph $\gamma\in\overline{\Gamma}$, which contains the three edges of a triangle or none of them, this graph does not have ``end points", as can be seen in Figure \ref{bombin3}, that is, it is a closed string. This subgraph defines a Pauli operator
\begin{equation}\label{stringsubsistem}
O_\gamma = \bigotimes_v P_v,
\end{equation}
with $P_v=I,X,Y,Z$, according to the color of the edge of the subgraph to which the vertex $v$ belongs.

\begin{figure}[!h]
	\centering
	\includegraphics[scale=0.9]{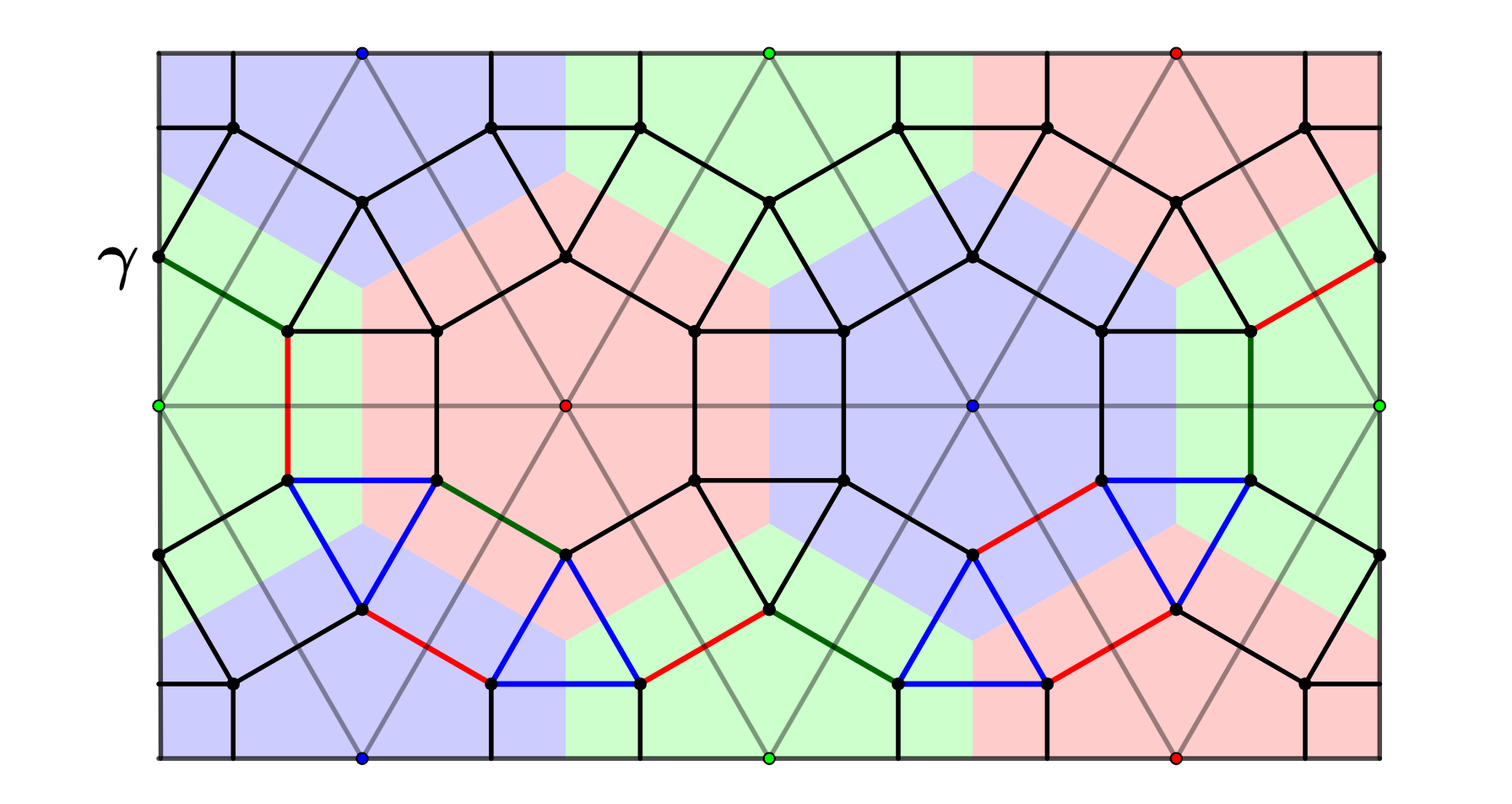} 
	\caption{Colored edges representing the $\gamma\in\overline{\Gamma}$ subgraph.} \label{bombin3}
\end{figure}

It follows that the operators $O_\gamma$ given in (\ref{stringsubsistem}) belong to $C(\G)$, and there is a bijection between these subgraphs and $C(\G)$.

The triangles of these subgraphs are connected in pairs, and if we look at the tessellation $\Gamma^*$, these triangles always connect vertices of the same color. This allows you to classify the strings according to color. As with color codes, strings of the same color commute, while strings of different colors anti-commute.

Define $S_v^c =O_\gamma$, with $c=r,g,b$, then $\gamma$ is the smallest colored string $c$ around the vertex $v\in V^*$ of the tessellation $\ Gamma^*$ (see Figure \ref{bombin4}). These operators are the stabilizer generators of $\textbf{S}$, that is, $\textbf{S}=\langle S_v^c;\ v\in V^*\ \textrm{e} \ c=r,g ,b\rangle$. However, these generators are not all independent, as the following relationships apply:

\begin{equation}\label{bombinrelacao}
\prod_c S_v^c = I\ \ \textrm{e}\ \ \prod_v S_v^c = I,
\end{equation}
where in the first case we vary the color and the vertex is fixed, and in the second we vary the vertex and the color is fixed. Therefore, the number of independent stabilizer generators (i.e., the rank of $\textbf{S}$) will be $s=2V^*-2$, where $V^*$ is the number of vertices of the dual tessellation $\Gamma^*$. The number of physical qubits $n$ will be equal to $n=3F^*$, where $F^*$ is the number of faces of the tessellation $\Gamma^*$. The number of encoded qubits is $k=2-\chi = 2g$. As $n = k + r + s$, it follows that the number of gauge qubits will be $r = 2F^* -\chi$. Therefore, the gauge qubits depend on the surface.

\begin{figure}[!h]
	\centering
	\includegraphics[scale=0.9]{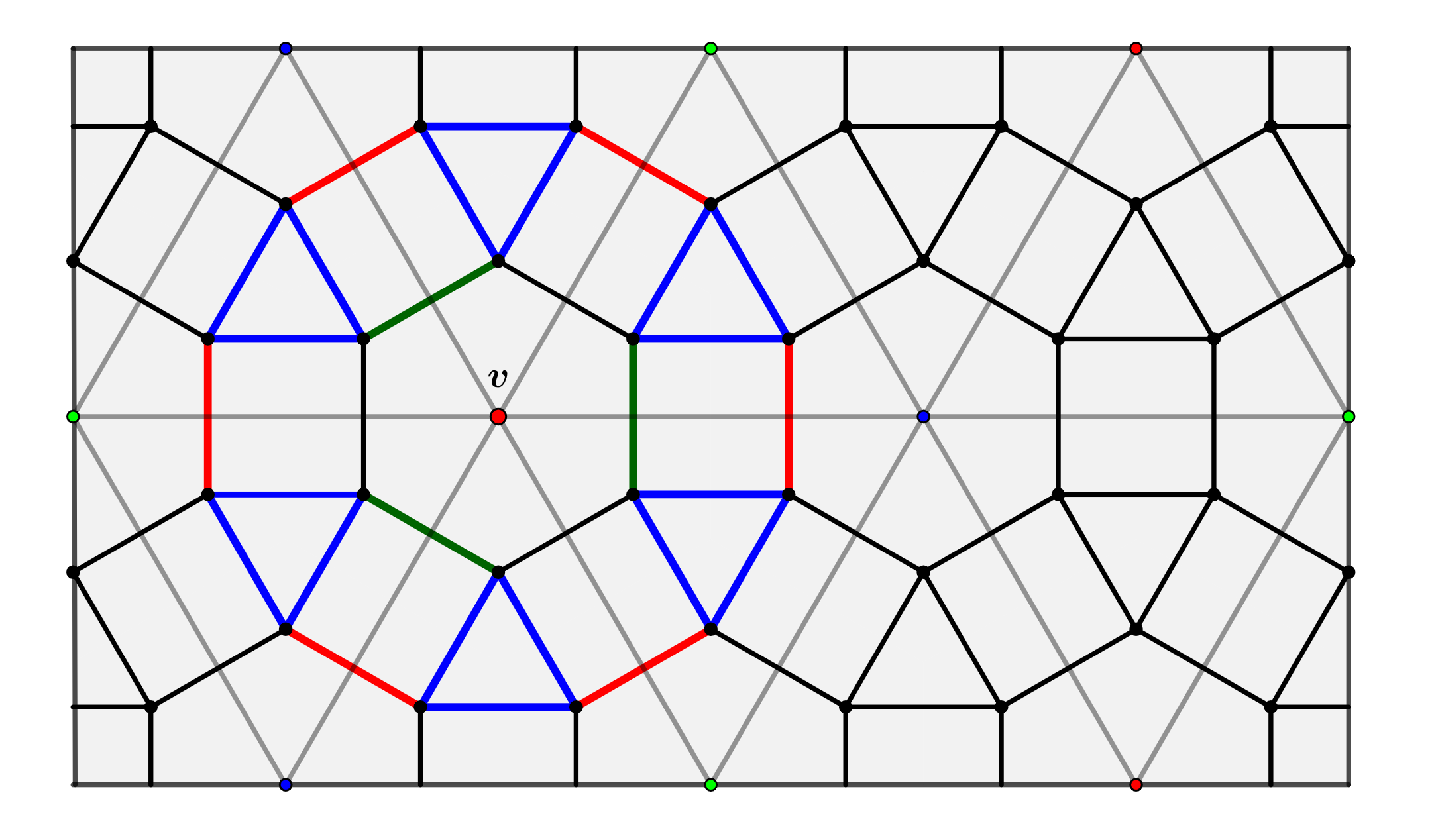} 
	\caption{Colored edges of the operator $S_v^g$.} \label{bombin4}
\end{figure}

For the distance $d$, in \cite{bombinsub} it is given that $d\leq d_T$, where $d_T$ is the minimum length, in the number of triangles, between non-trivial closed strings. This is justified, since given a string operator $O_{\gamma}\in C(\G)$, and defining the set $E'$ as being the set of red, green, and only one blue edge of each triangle of $\ gamma$, we have to
\[
G:=\prod_{e'\in E'}\overline{K}_{e'}\in \G,
\]
and the weight of the operator $O_\gamma G$ is equal to the number of triangles in $O_\gamma$. It also provides a lower bound on the distance, which is given by $d\geq d_L$, where $d_L$ is the minimum length in terms of the number of edges, between the non-trivial homologically closed loops in the tessellation $\Gamma^*$.

We also consider the topological subsystem codes built in \cite{terhal}. To build these codes, the authors used trivalent hypergraphs $\Gamma_h$ embedded in a plane or a torus. The qubits are fixed at the vertices, and the set of all vertices is denoted by $V$. Edges are of two types: edges that connect two vertices, which are called link or rank-$2$ edges, and edges that connect three vertices, which are called triangles or rank-$3$ edges. The sets of all rank-$2$ and rank-$3$ edges are denoted by $E_2$ and $E_3$, respectively, therefore, $E=E_2\cup E_3$.

\begin{defi}
An edge $e'$ is incident to a vertex $u$ if there exists $v\in V$ such that $e'=(u,v)$ or if there exist $v,w\in V$ such that $e' =(u,v,w)$.
\end{defi}

There are four restrictions imposed in \cite{terhal} for such hypergraphs and one more imposed in \cite{sarvepalli}. They are:
\begin{itemize}
\item[($H_1$)] $\Gamma_h$ has only rank-$2$ or rank-$3$ edges;
\item[($H_2$)] Each vertex has exactly three incident edges;
\item[($H_3$)] Two distinct edges intersect at at most one vertex;
\item[($H_4$)] Rank-$3$ edges are two by two disjoint;
\item[($H_5$)] The edges of the hypergraph $\Gamma_h$ must be $3$-colorable.
\end{itemize}

As in \cite{bombinsub}, in \cite{terhal} edge operators are also defined, but unlike there, there are two types here. The link operators $K_{e'}$ where $e'\in E_2$, and the triangle operators $K_{e'}$ where $e'\in E_3$. By convention, given an edge $e'=(u,v)\in E_2$, we have that $K_{e'}:=X_uX_v$ or $K_{e'}:=Y_uY_v$. If $e'=(u,v,w)\in E_3$, we have that $K_{e'}:=Z_uZ_vZ_w$. This choice does not change the code in any way, as a depolarizing channel is being assumed, i.e, errors of the types $X$, $Y$, and $Z$ occur with the same probability.

The next result is from \cite{breuck}. 

\begin{Theo}\label{teoregradecomutacao}
Let $e',e''\in E$ be any two edges of $\Gamma_h$. So, for their respective operators, $K_{e'}$ and $K_{e''}$, we have
\begin{equation}\label{regradecomutacao}
K_{e'}K_{e''} = (-1)^{\eta(e',e'')} K_{e''}K_{e'},
\end{equation}
with $\eta(e',e'') = 0$, if $e'$ and $e''$ share an even number of vertices or $e',e''\in E_3$, and otherwise $ \eta(e',e'') = 1$.
\end{Theo} 

\begin{Rema}
To obtain a subsystem code from the hypergraph $\Gamma_h$, it is essential that the commutation rule be satisfied for all edge operators of the hypergraph $\Gamma_h$. Therefore, in \cite{sarvepalli}, the authors added the restriction $(H_5)$ to the hypergraph $\Gamma_h$. Only $(H_1),\ldots,(H_4)$ does not guarantee the existence of subsystem codes. For instance, as can be seen in \cite{sarvepalli}, the Petersen graph satisfies $(H_1),\ldots,(H_4)$ but does not admit defining a subsystem code.
\end{Rema}

We will now define a closed hypercycle in a hypergraph. Roughly speaking, restricting to the cases covered in \cite{bombinsub}, these closed hypercycles are the strings seen previously.

\begin{defi}
We call a subset of edges $M\subset E$ a closed hypercycle, or just a hypercycle, if every vertex of the hypergraph has an even number of edges incident to $M$.
\end{defi}

Using these hypercycles, we define an operator that will be very important, called the loop operator.
\begin{defi}
Let $M\subset E$ be any closed hypercycle. The operator
\begin{equation}
W(M) = \prod_{e'\in M} K_{e'}
\end{equation}  
is the tensor product between the link operators and the triangle operators acting on the edges of $M$. We call $W(M)$ the loop operator.
\end{defi}

The next results tell us when an edge operator (link or triangle) commutes with a loop operator, and when two loop operators commute. Their proof can be seen at \cite{breuck} and \cite{terhal}.

\begin{Theo}
Let $e'\in E$ be any edge and $M\subset E$ be any closed hypercycle. Their corresponding operators commute if and only if $e'$ is not a triangle contained in $M$, that is, 
\[
[W(M), K_{e'}] = 0 \Leftrightarrow e'\notin M \cap E_3 \,.
\]
\end{Theo}

\begin{cor}\label{coroloperadoresloops}
Let $M,M'\subset E$ be any two closed hypercycles, then
\[
W(M)W(M') = (-1)^{\eta(M,M')} W(M')W(M) \,,
\]
where $\eta(M,M'):=|M\cap M'\cap E_3|$ is the number of triangles shared by $M$ and $M'$.
\end{cor}

\begin{Rema}
It follows from the Corollary $\ref{coroloperadoresloops}$ that two loop operators will commute if they have an even number of triangles in common, and anticommute if they have an odd number of triangles in common. Thus, the number of triangles in any closed hypercycle must be even for $W(M)$ to commute with itself.
\end{Rema}

We also define the group
\begin{equation}\label{gloop}
\G_{loop}:=\{W(M); M\subseteq E\ \textrm{is a closed hypercycle}\}.
\end{equation}

Through this group, we can obtain the gauge group $\G$ and the stabilizer group $\textbf{S}=\G\cap C(\G)$ using the identities

\begin{equation}
\G=C(\G_{loop})\ \ \textrm{e}\ \ \textbf{S}=\G_{loop}\cap C(\G_{loop}).
\end{equation}
Thus, it follows that $C(\G) = \G_{loop}$.

The stabilizer operators will be the loop operators that commute with all the other loop operators. The bare logical operators $\overline{X}_i,\overline{Z}_i\in C(\G)$ can be chosen as a pair of loop operators satisfying the standard commutation rules of Pauli operators.

To prove that $C(\G)=\G_{loop}$ and the other results, it is necessary to transform the hypergraph $\Gamma_h$ into an ordinary graph, denoted by $\overline{\Gamma}_h$. Firstly, the set of vertices $\overline{V}$, that is, of qubits of $\overline{\Gamma}_h$, will be the same as that of $\Gamma_h$. The edges of $\overline{\Gamma}_h$ will all be rank-$2$, so the links, that is, the edges of $E_2$, will continue to be the same links in $\overline{\Gamma}_h$. Triangles will be transformed into links as follows: Given a triangle $(u,v,w)$, we will consider $(u,v)$, $(u,w)$ and $(v,w)$ as links, that is, each triangle contributes three new links in $\overline{\Gamma}_h$. Thus, $\overline{\Gamma}_h$ is obtained from $\Gamma_h$ just by transforming each triangle into three links.

In this way, the triangle operator $K_{e'}=Z_uZ_vZ_w$ will be divided into three operators: $\overline{K}_{u,v}=Z_uZ_v$, $\overline{K}_{u,w}=Z_uZ_w$, and $\overline{K}_{v,w}=Z_vZ_w$. Note that these three operators are not independent, as $\overline{K}_{u,v}\overline{K}_{u,w}\overline{K}_{v,w}=(Z_uZ_v)(Z_uZ_w )(Z_vZ_w) = I$. Thus, denoting the link operators $K_{e'}$ also by $\overline{K}_{e'}$, we have that the gauge group $\G$ follows as in (\ref{grupogauge}), that is, $\G = \langle \overline{K}_{e'};\ e'\in \overline{E}\rangle$, where $\overline{E}$ are the edges of $\overline{\Gamma }_h$.

The proof of the next result can be seen in \cite{terhal}.
\begin{lema}
Let $\G$ be the group generated by the link operators $\overline{K}_{e'}$, where $e'\in \overline{E}$ is an edge of $\overline{\Gamma}_h$. So $\G_{loop} = C(\G)$.
\end{lema}

The next result, which can also be seen in \cite{breuck, terhal}, is one of the most important results, as it tells us when ($C_2$) is satisfied. Remembering that the item ($C_2$) says that the extraction of the syndrome can be done by measuring the eigenvalues of the link operators, which are $2$-local, that is, measuring only the elements of the gauge group $\G$.

\begin{Theo}\label{teosindrome}
The eigenstate of a stabilizer $S\in\textbf{S}$ can be measured by a set of link operators if and only if $S$ can be written as a product of these link operators, i.e.
\begin{equation}
S=\overline{K}_m\cdots \overline{K}_1
\end{equation}
and $\overline{K}_j$ commutes with the ordered product of all preceding link operators, that is,
\begin{equation}\label{regrasindrome}
[\overline{K}_j,\overline{K}_{j-1}\cdots \overline{K}_1]=0\ \forall j\in\{2,\ldots,m\}.
\end{equation}
\end{Theo}

The way of writing the stabilizer $S$ does not matter and is not necessarily unique, but it must satisfy (\ref{regrasindrome}).

Although this construction of subsystem codes using the group $\G_{loop}$ made in \cite{terhal} provides subsystem codes, these codes are not necessarily topological. In \cite{sarvepalli}, the authors give a counterexample, where a stabilizer is obtained that is associated with a homological non-trivial closed hypercycle, that is, it will not satisfy the condition $(C_1)$ of the definition \ref{defisubtopolo}. Even satisfying $(C_2)$ and $(C_3)$, it will not be a topological subsystem code.

Fortunately, this is not a restriction to be subsystem code, it just won't be topological. Thus, to use this construction, we will need more properties to guarantee that they will actually give topological subsystem codes.

Obtaining the parameters of topological subsystem codes is not an easy task, as they depend on the type of hypergraph and the surface where the hypergraph is embedded. 

We will now consider the topological subsystem codes built in \cite{sarvepalli}. In this construction, the authors start from a trivalent and 3-colorable graph, denoted by $\Gamma_2$, imposing the restriction that there is a non-empty set of faces whose number of edges is a multiple of and greater than $4$.

From this graph, a trivalent hypergraph $\Gamma_h$ is constructed, satisfying the conditions $(H_1),\ldots, (H_5)$ seen previously, thus, it is guaranteed that the codes obtained are subsystem codes, which are described by $\G_{loop}$ as in \cite{terhal}.

We will now describe how the process of constructing the hypergraph $\Gamma_h$ is carried out. Consider a trivalent and $3$-colorable graph $\Gamma_2$ with the requirement that there be a non-empty set of faces as previously stated. Let $F_R$, $F_B$, and $F_G$ be the sets of all red, blue, and green faces, respectively. Suppose without loss of generality that the faces $f\in F$, where $F\subseteq F_R$, are such that the number of sides of $f$ is congruent to $4(\textrm{mod}\ 0)$ and greater than $4$. For each $f\in F$, a face $f'$ is added inside $f$, so that $f'$ has half the number of edges of $f$.

Coloring the edges of $f$ alternately using the colors blue and green, there are two ways to construct the edges of rank-$3$ (the triangles) using the vertices of the face $f'$ inside $f$. The first way is to consider triangles using a vertex of $f'$ and a blue edge of $f$, as in Figure \ref{saverpalli1}-($a$), so that these triangles do not intersect, that is, they are disjoint. The second way is analogous to the first, but now the green edges are considered instead of the blue ones, as in Figure \ref{saverpalli1}-($b$).

\begin{figure}[!h]
	\centering
	\includegraphics[scale=0.9]{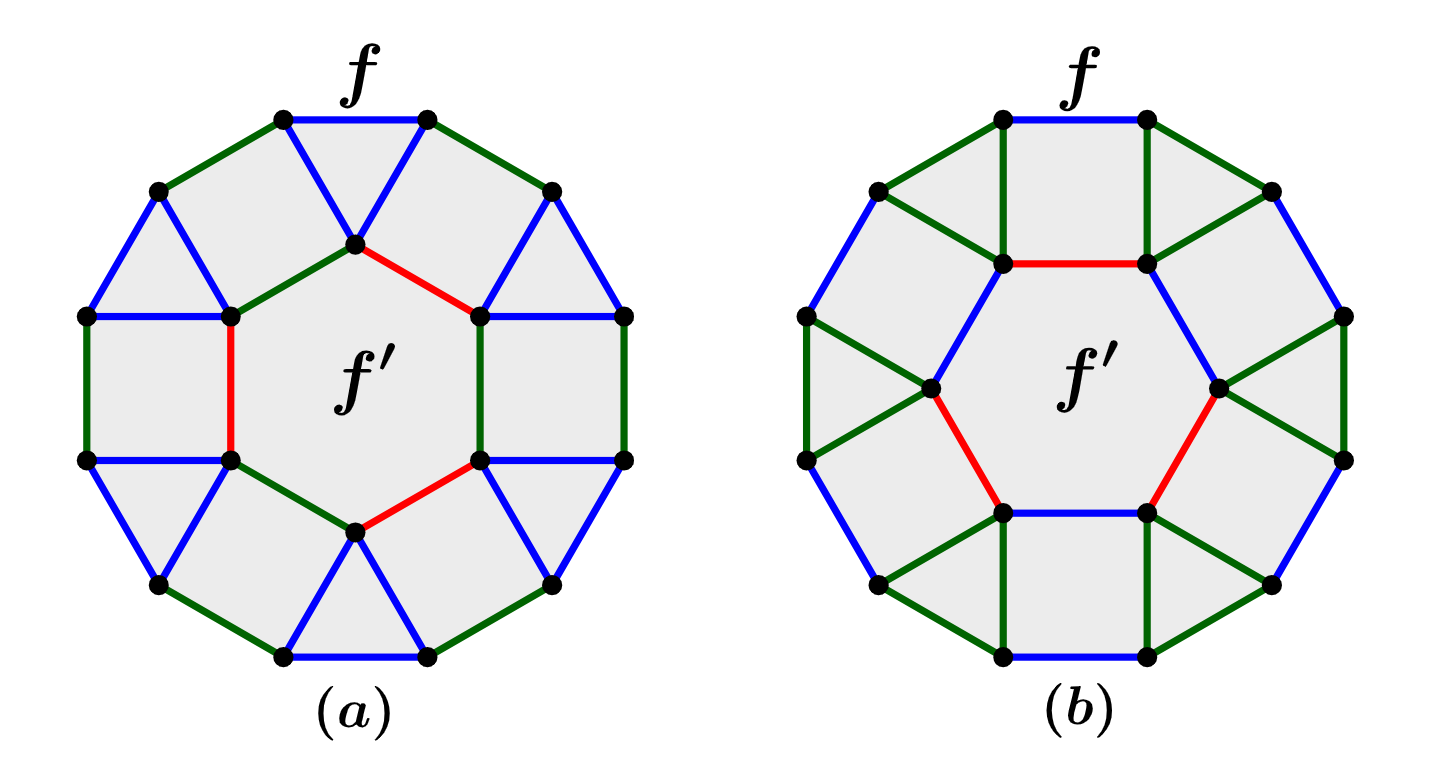} 
	\caption{Rank-$3$ edges inserted into face $f\in F$ from the vertices of face $f'$. ($a$) blue rank-$3$ edges and ($b$) green rank-$3$ edges.} \label{saverpalli1}
\end{figure}

The rank-$3$ edges have the same color as the edge $f$ used in its construction. The edges of $f'$ are of different colors from the color of the edges of rank-$3$, as in Figures \ref{saverpalli1}-($a$) and \ref{saverpalli1}-($b$). The resulting graph is the desired $\Gamma_h$ hypergraph.

In \cite{sarvepalli}, the authors proved that the graph $\Gamma_h$ satisfies $(H_1),\ldots,(H_5)$. Therefore, it gives rise to subsystem codes, whose group of loop operators is given as in (\ref{gloop}) and the gauge group is given by $\G=C(\G_{loop})$.

Because $F\subseteq F_R$ is a subset, the stabilizers of the code vary depending on $F$. However, given any face $f\in F$, there are two independent hypercycles in $\Gamma_h$ that we can associate with this face, and consequently, two independent stabilizing generators. The first of them, denoted by $f_{\sigma_1}$, is composed only of the edges of $f'$, as in Figure \ref{saverpalli2}-($a$). The second, denoted by $f_{\sigma_2}$, is formed by the triangles inserted in $f$, together with the rank-$2$ edges on the border of $f$ and the edges of $f'$ of the same color as the edges of $f$, as in Figure \ref{saverpalli2}-($b$). In \cite{sarvepalli} it is proved that the operators $W(f_{\sigma_1})$ and $W(f_{\sigma_2})$ associated with these hypercycles are independent stabilizer generators, and satisfy the Theorem \ref{teosindrome}.

It is possible to obtain a third hypercycle $f_{\sigma_3}$ (see Figure \ref{saverpalli2}-($c$)) and, consequently, a third stabilizer generator $W(f_{\sigma_3})$, doing the modulo 2 sum of $f_{\sigma_1}$ and $f_{\sigma_2}$, and multiplying $W(f_{\sigma_1})$ by $W(f_{\sigma_2})$, respectively.

\begin{figure}[!h]
	\centering
	\includegraphics[scale=0.9]{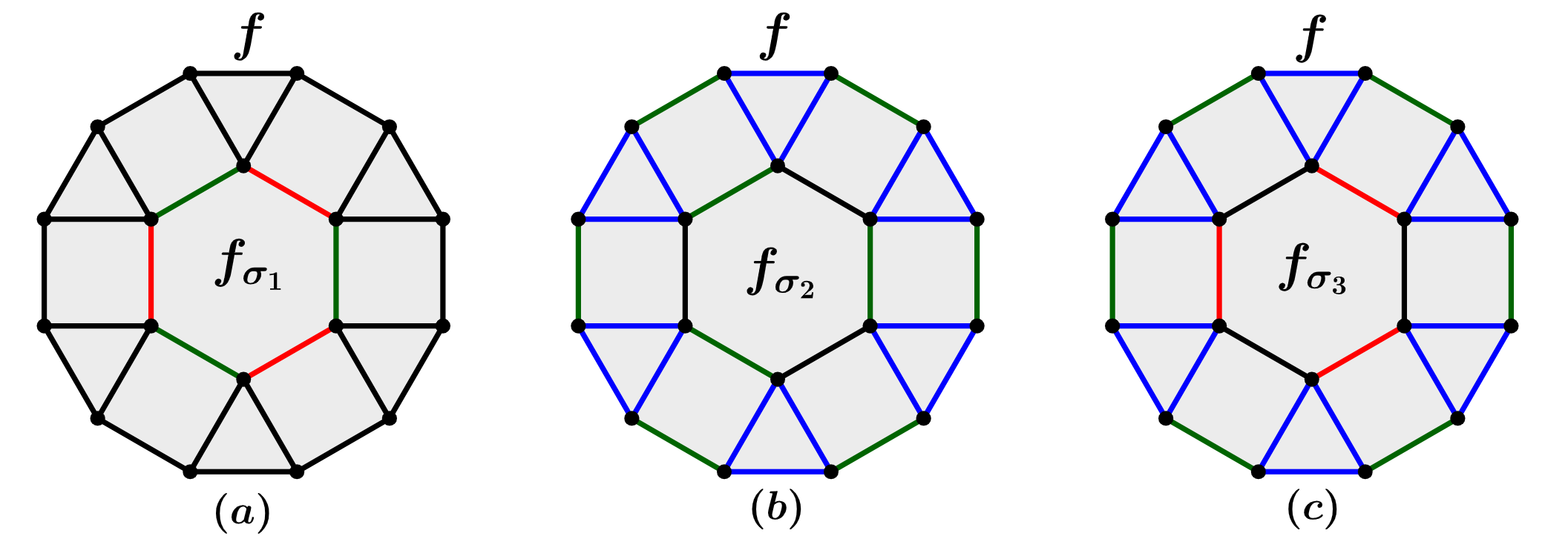} 
	\caption{Hypercycles corresponding to face $f\in F\subseteq F_R$ represented by colored edges. ($a$) Hypercycle $f_{\sigma_1}$, ($b$) hypercycle $f_{\sigma_2}$ and ($c$) Hypercycle $f_{\sigma_3}$.} \label{saverpalli2}
\end{figure}

Note that the operators $W(f_{\sigma_1})$ and $W(f_{\sigma_2})$, when restricted to the cases covered in \cite{bombinsub}, coincide with the operators $S_{v}^{c }$ seen previously.

For distance, the authors in \cite{sarvepalli} proved that $d$ will be upper bound by the smallest number of triangles in a homological non-trivial closed hypercycle, as done in \cite{bombinsub}, that is, it will be the smallest weight among the operators of $\G_{loop}\setminus \textbf{S}$, that is, $C(\G)\setminus \textbf{S}$.

Before presenting the new families of codes, where it will be possible to obtain the parameters, we need to present a way to construct a trivalent and $3$-colorable graph $\Gamma_2$ from another graph embedded in a surface. This construction is due to \cite{17sarvepalli}, but it can also be seen in \cite{sarvepalli} and is used to build one of the families that we will present below.

Given a graph $\Gamma$ embedded in a surface, one of the colors is chosen to color all its faces. Each edge of the graph is transformed into two edges, and the new faces obtained are now all colored with the same color but a different color from the color already chosen, as can be seen in Figure \ref{saverpalli3}. Finally, each vertex of valence $j$ is transformed into a face of $j$ edges. These new faces are all colored with the color that has not yet been used, and, therefore, a trivalent and $3$-colorable graph is obtained, which is denoted by $\Gamma_2$.

The next result from \cite{sarvepalli} tells us how to construct the first family of codes made by the authors, along with the parameters of these codes. 

\begin{figure}[!h]
	\centering
	\includegraphics[scale=0.9]{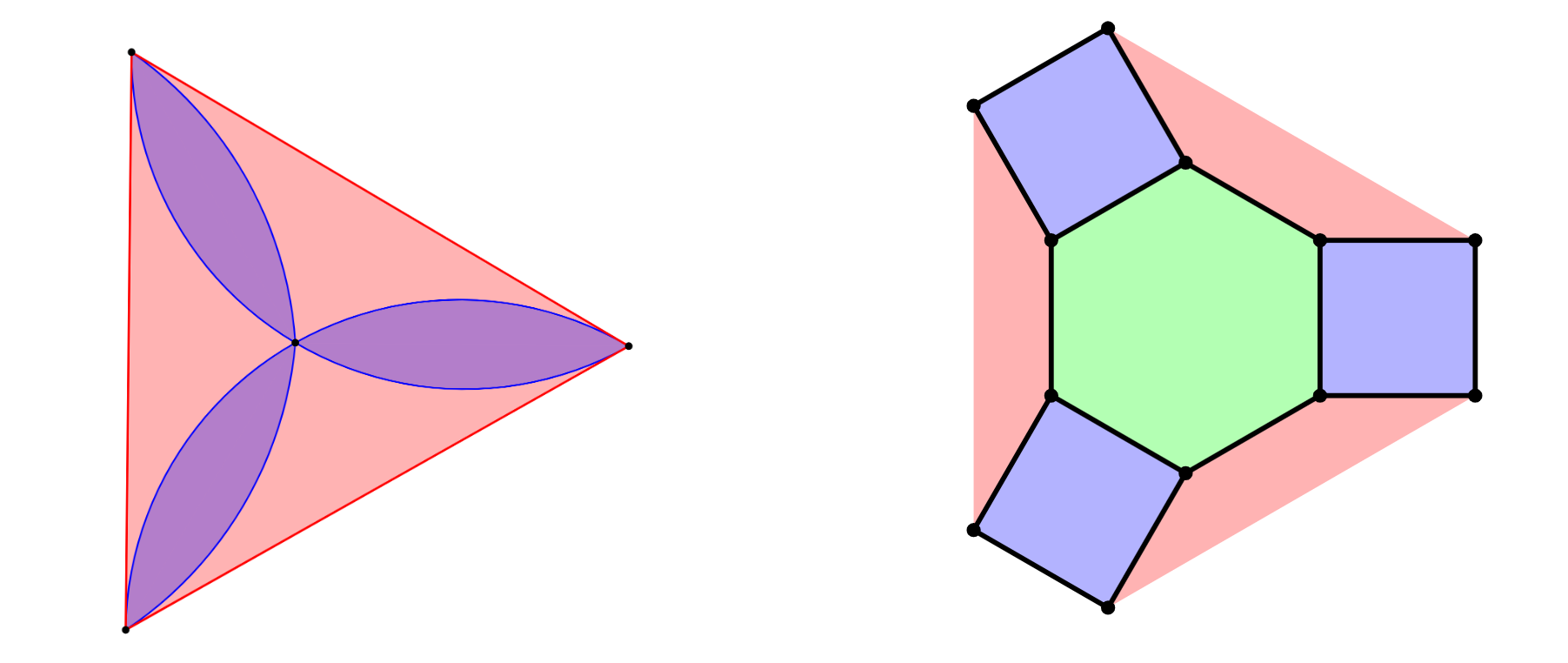} 
	\caption{Transforming any graph $\Gamma$ into a trivalent and $3$-colorable graph $\Gamma_2$.} \label{saverpalli3}
\end{figure}

We need the definition of a bipartite graph.
\begin{defi}
A graph $\Gamma=(V,E)$ is called bipartite, if $V$ admits a partition into two sets such that every edge has its endpoints in different sets. Vertices in the same set cannot be adjacent.
\end{defi}

\begin{Theo}\label{teo5indiano}
Let $\Gamma$ be a graph such that every vertex has an even degree (valence) greater than two. Construct a trivalent, $3$-colorable graph $\Gamma_2$ using the construction described previously. Then, using $\Gamma_2$ to construct the hypergraph $\Gamma_h$, where $F$ will be the set of $v$ faces of $\Gamma_2$, and the edges of rank-$3$ are the boundaries of $e$-faces of $\Gamma_2$. If $l$ is the number of edges of rank-$3$ of a homologically non-trivial hypercycle, then we obtain a subsystem code with parameters
\begin{equation}\label{teo5sarvepalli}
[[6e,1+\delta_{\Gamma^*,\textrm{bipartido}}-\chi,4e-\chi,d\leq l]],
\end{equation}
where $e=|O(\Gamma)|$ and $\delta{\Gamma^*,\textrm{bipartite}} = 1$ if $\Gamma^*$ is bipartite and $0$ otherwise.
\end{Theo}

Before presenting the next result that we will use to construct our second family of codes, we need the definition of a medial graph. The proof of the result can also be found in \cite{sarvepalli}.

\begin{defi}
We call the medial graph of a graph $\Gamma$ the graph that is obtained by placing a vertex on each edge of $\Gamma$ and an edge connecting two of these vertices, if the edges associated with these vertices in $\Gamma$ are adjacent. We denote the medial graph of $\Gamma$ by $\Gamma_m$.
\end{defi}

\begin{Theo}\label{teo6indiano}
Let $\Gamma$ be a graph whose vertices have even valences greater than $2$ and $\Gamma_m$ be its medial graph. Construct the trivalent and $3$-colorable graph $\Gamma_2$ from $\Gamma_m^*$ (dual of $\Gamma_m)$ using the construction presented previously. If $\Gamma_m^*$ is bipartite, the set of $v$-faces of $\Gamma_2$, denoted by $F_v$, forms a bipartition $F_v\cup F_f$, where $|F_v|=|V(\ Gamma)|$. Making the construction of $\Gamma_h$ with $F_v\subset F$ such that the edges of rank-$3$ are not on the borders of the $e$-faces of $\Gamma_2$. Let $l$ be the number of edges of rank-$3$ in a homologically non-trivial hypercycle, then we obtain a subsystem code with parameters
\begin{equation}
[[10e,1+\delta_{\Gamma^*,\textrm{bipartido}}-\chi,6e-\chi,d\leq l]],
\end{equation}
where $e=|O(\Gamma)|$ and $\delta{\Gamma^*,\textrm{bipartite}} = 1$ if $\Gamma^*$ is bipartite and $0$ otherwise.
\end{Theo}

In \cite{sarvepalli}, the authors proved that these codes cannot be obtained by the construction made in \cite{bombinsub}. In particular, all the codes of Theorem \ref{teo6indiano} are different from those obtained in \cite{bombinsub}, and the codes of Theorem \ref{teo5indiano} are different when the graph $\Gamma^*$ is non-bipartite. They also proved that these codes satisfy the $(C_2)$ condition and that they are topological subsystem codes.

\section{A new construction to obtain topological subsystem codes }\label{construcaomeussubcode}

Inspired by the constructions of topological subsystem codes in the last section, in this section we present a generalization of the construction given in \cite{sarvepalli}, where we will construct new subsystem codes and present four new families of topological subsystem codes. We will prove that these codes are topological subsystem codes.

In our construction, instead of starting from a trivalent and $3$-colorable graph, we will start from a trivalent and $3$-colorable tessellation. Although the graph approach is more generic, using tessellations, we will be able to obtain greater control over the parameters of the codes obtained. We also draw attention to the fact that in \cite{artigo} all trivalent and $3$-colorable tessellations existing on the surfaces were determined, as well as the quantities of faces, edges, and vertices of these tessellations were calculated, the which we will use here.

Consider a trivalent and $3$-colorable tessellation $\{2p_1,2p_2,2p_3\}$. We denote the set of red, green, and blue faces by $F_R$, $F_G$, and $F_B$, respectively. In order to make our construction easier, we will fix the colors of the tessellation polygons as follows: Suppose without loss of generality that the $2p_1$-gons are red, the $2p_2$-gons are green, and the $2p_3$-gons are blue.

Consider $F\subseteq F_R$ such that $F\neq\emptyset$ and assume that $2p_1>4$ with $p_1>2$ odd. If we make the construction proposed in \cite{sarvepalli}, only for trivalent and 3-colorable tessellations $\{2p_1,2p_2,2p_3\}$, we will have the restriction that $2p_1>4$ with $p_1>2$ even, thus we will only work with the case $p_1>2$ odd.

To construct the hypergraph $\Gamma_h$, we insert a face $f'$ with $p_1$ sides inside each face $f\in F$ and construct the triangles from the vertices of $f'$, so that one of its edges is also an edge of a blue $2p_3$-gon, and these triangles do not intersect each other, as in Figure \ref{fig6-6-6-subsistem}.

\begin{figure}[!h]
	\centering
	\includegraphics[scale=1.1]{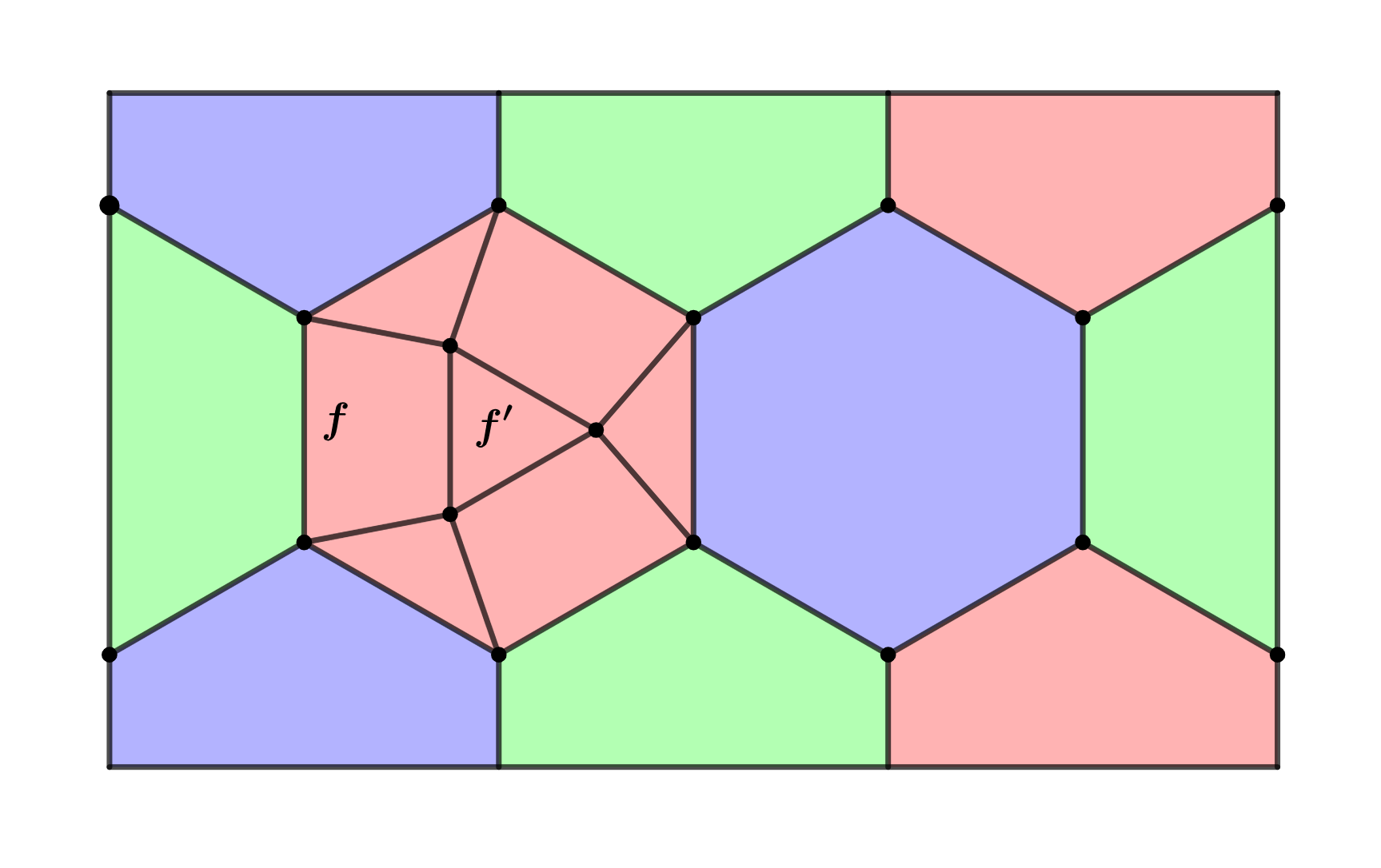} 
	\caption{Tessellation $\{6,6,6\}$, where inside the hexagonal face $f\in F_R$ we insert the triangular face $f'$, and from the vertices of $f'$ we construct three triangles.} \label{fig6-6-6-subsistem}
\end{figure}

By fixing this choice in relation to the construction of triangles together with fixing the colors, we are also fixing an order for the tessellations, as tessellation $\{6,12,4\}$ will provide a different hypergraph than tessellation $\{6,4,12\}$, as can be seen in Figure \ref{quasegrafos(6,12,4 e 6,4,12)}.

\begin{Rema}
If $p_1$ is even, to obtain the construction done in \cite{sarvepalli},just color the edges as we saw in Section $\ref{codigosubsistemtopolo}$. However, it is worth highlighting that in our approach, we are starting from a tessellation and not from a graph, as in \cite{sarvepalli}. Furthermore, if $p_1$ is even, $p_3=2$, and $F=F_R$, then we are in the conditions of Theorem $\ref{teo5indiano}$, as we have $\{2p_1,2p_2,4\}$. This is the only family of tessellations to which this theorem applies.
\end{Rema}

\begin{figure}[!h]
	\centering
	\includegraphics[scale=0.9]{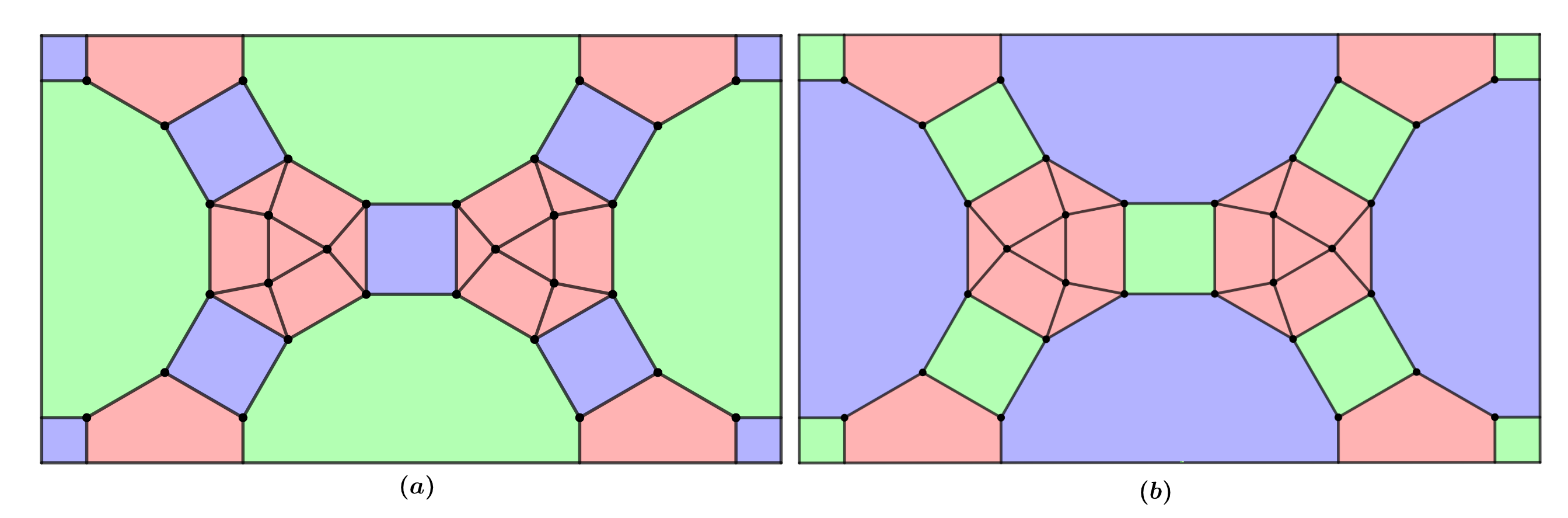} 
	\caption{Triangles inserted inside two faces of $F_R$ in the tessellations (a) $\{6,12,4\}$ and (b) $\{6,4,12\}$.} \label{quasegrafos(6,12,4 e 6,4,12)}
\end{figure}

%Repare que até o momento, nossa contrução é basicamente a mesma que feita em \cite{sarvepalli}, com a diferença de que estamos trabalhando somente com tesselações trivalentes e 3-coloríveis e estamos fixando $p_1$ ímpar, e é essa escolha que nos fornecerá os novos códigos.

For this construction with the odd case $p_1$, the resulting hypergraph does not satisfy the restriction $(H_5)$, therefore, it does not satisfy the switching rule and, therefore, does not admit the construction of subsystem codes. To correct this problem, we choose any three distinct edges of $f'$ and take a point on each of them, different from the extremes. Then we connect these points to obtain a triangle, as in Figure \ref{face10-10-4}. The edges where we chose the points give rise to two new edges, and these points are new vertices. Therefore, the face $f'$ will have $p_1 + 3$ edges.

\begin{figure}[!h]
	\centering
	\includegraphics[scale=5.5]{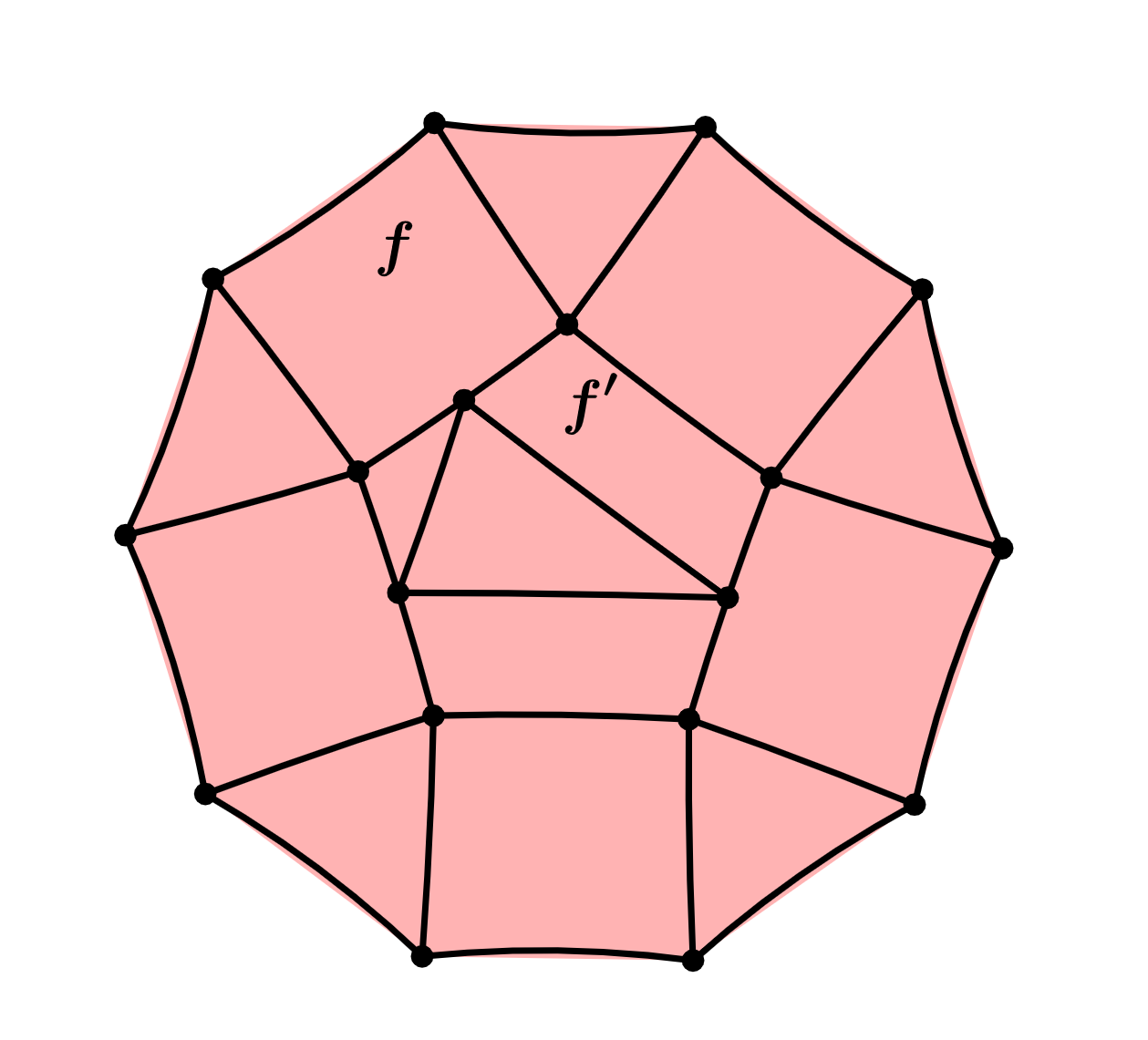} 
	\caption{Face $f\in F_R$ of the hyperbolic tessellation $\{10,10,4\}$, where we construct a triangle inside the face $f'$. Thus, $f'$ now has 8 edges.} \label{face10-10-4}
\end{figure}

We color all the edges of the triangles blue, and the other edges of $f$, where the triangles were inserted, we color red. For possible faces $f\in F_R$ such as $f\not\in F$, we alternate between red and blue, where the blue edges are those connected to the $2p_3$-gons. For the $p_1 + 3$ edges of $f'$, we alternate between the colors red and green, and the remaining edges, which belong to the $2p_2$-gons and $2p_3$-gons, we color green, as in Figure \ref{hipergrafo6-12-4}.

\begin{figure}[!h]
	\centering
	\includegraphics[scale=0.9]{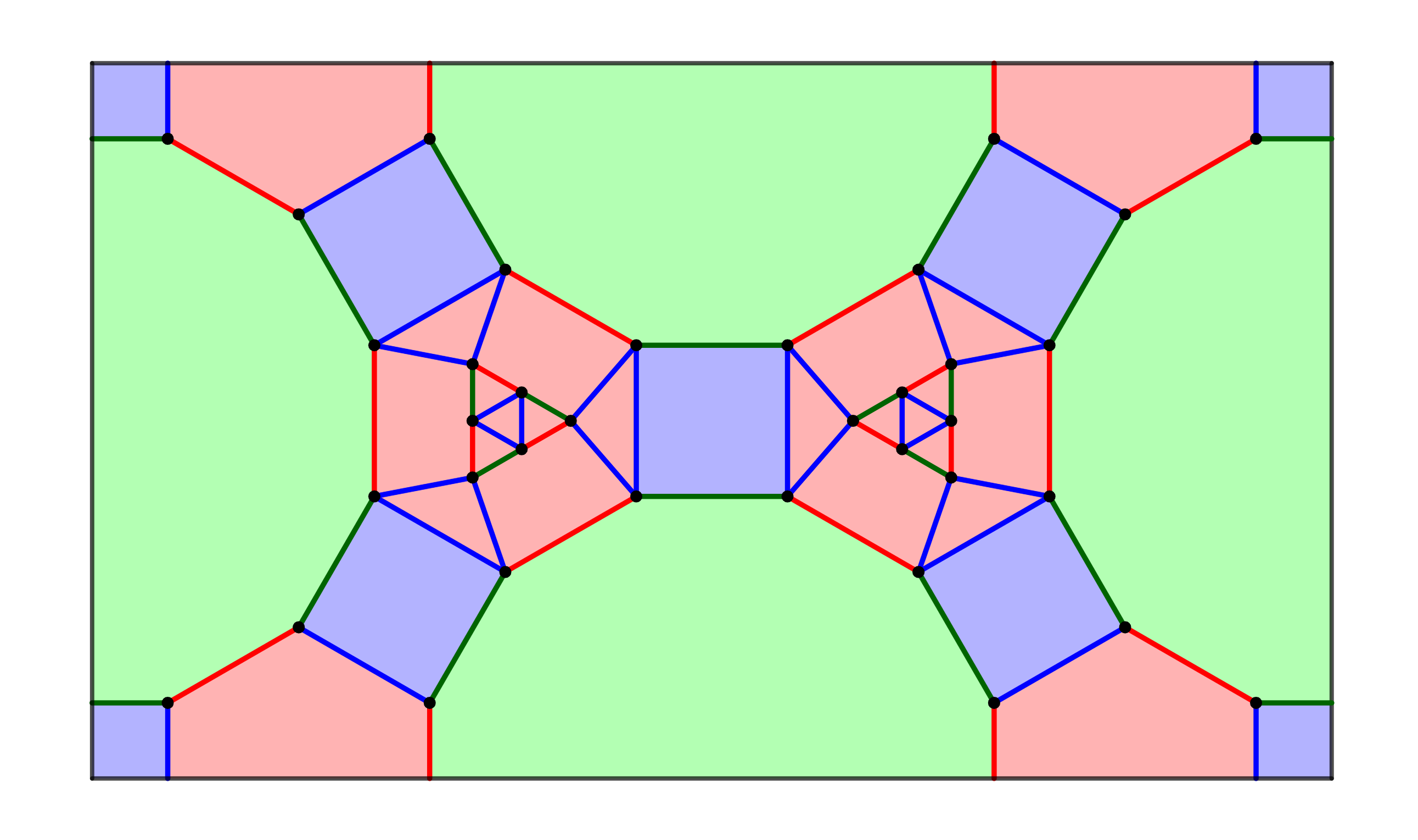} 
	\caption{Hypergraph $\Gamma_h$ obtained from tessellation $\{6,12,4\}$.} \label{hipergrafo6-12-4}
\end{figure}

Note that when we construct the triangle from points taken on the edges of $f'$, we correct the problem of the hypergraph not satisfying $(H_5)$. We denote the resulting hypergraph by $\Gamma_h$ and use it to construct the subsystem codes in the same way as in \cite{terhal}.

To see that the hypergraph $\Gamma_h$ given in our construction can be used to construct subsystem codes as in \cite{terhal}, we will prove that $\Gamma_h$ satisfies $(H_1),\ldots,(H_5)$. 

\begin{Theo}\label{satisfazasrestricoes}
The hypergraph $\Gamma_h$ obtained by the previous construction, starting from a tessellation $\{2p_1,2p_2,2p_3\}$ with odd $p_1>2$, satisfies the constraints $(H_1),\ldots,(H_5)$, and, therefore, it gives rise to subsystem codes whose group formed by loop operators is given by
\begin{equation}
\G_{loop} = \{W(M);\ M\subseteq E\ \ \textrm{é um hiperciclo fechado}\},
\end{equation}
where $E$ is the set formed by the edges of rank-$2$ and rank-$3$, and the gauge group $\G$ is given by $\G = C(\G_{loop})$.
\end{Theo}

\begin{proof}
In the tessellation $\{2p_1,2p_2,2p_3\}$, we only have rank-$2$ edges. When we insert $f'$ into $f\in F\subseteq F_R$ and create triangles from the vertices of $f'$, we are producing only rank-$2$ and rank-$3$ edges. In the same way, when we create the triangle inside $f'$ from points taken on three of its edges, we are producing one edge of rank-$3$ and three of rank-$2$. Therefore, $\Gamma_h$ satisfies $(H_1)$.

It also follows from the way we insert the faces $f'$ and create the triangles that the original vertices of $\{2p_1,2p_2,2p_3\}$ will continue to be trivalent. The vertices of $f'$ are also trivalent because of the triangles created from their vertices. Likewise, the three new vertices will be trivalent, as they divide an edge into two and have the new triangle as the third edge. Thus, $\Gamma_h$ satisfies $(H_2)$.

The restriction $(H_3)$ is satisfied because, despite creating rank-$2$ and rank-$3$ edges, it follows from the way they were created that any two distinct edges intersect at most one vertex.

The restriction $(H_4)$ is also satisfied, because when we create the triangles from $f'$, we already require that they do not intersect. And the last triangle is also created without intersecting with the others. Therefore, the rank-$3$ edges of $\Gamma_h$ are two-by-two disjoint.

Finally, $(H_5)$ is satisfied, as we constructed the triangle inside $f'$ so that the commutation rule is valid, that is, so that the edges of $\Gamma_h$ were $3$-colorable. Therefore, $\Gamma_h$ gives rise to subsystem codes as constructed in \cite{terhal}.
\end{proof}

\begin{Rema}
The choice of the three edges of $f'$ that will have the vertices of the triangle is made, in most cases, randomly, that is, we have no restrictions. However, we will have a small restriction for two families of tessellations. The first of them is for tessellation $\{2p_1,2p_2,4\}$ with $F=F_R$, and the second of them is tessellation $\{2p_1,4,6\}$ also with $F=F_R$. This restriction is due to the fact that the Theorem $\ref{teosindrome}$ has to be satisfied to guarantee that $(C_2)$ is valid. We will study these two cases in particular in the next section. For any other family of tessellations, we can take the three edges of $f'$ at random.
\end{Rema}

As we saw in Section \ref{cap5}, we fixed the triangle operators $K_{e'}=Z_uZ_vZ_w$, where $e'=(u,v,w)$. When we transform $\Gamma_h$ into the ordinary graph $\overline{\Gamma}_h$, we will have $\overline{K}_{u,v}=Z_uZ_v$, $\overline{K}_{u,w}=Z_uZ_w $ and $\overline{K}_{v,w}=Z_vZ_w$. Thus, in $\overline{\Gamma}_h$ given $e'=(u,v)\in\overline{E}$, we have that $\overline{K}_{e'} = Z_uZ_v$ if $e '\in\overline{E}_B$, $\overline{K}_{e'} = X_uX_v$ if $e'\in\overline{E}_R$ and $\overline{K}_{e'} = Y_uY_v$ if $e'\in\overline{E}_G$, where $\overline{E}_B$, $\overline{E}_R$ and $\overline{E}_G$ are the sets with all blue, red and green edges, respectively, of the graph $\overline{\Gamma}_h$.

As the non-empty set $F\subseteq F_R$ is any set, we cannot generally determine all the stabilizers in the code. However, for each $f\in F$, we can always determine three stabilizer generators, where only two are independent, and for each $2p_2$-gon, we can determine one more independent stabilizer generator. We show this in the following lemma.

\begin{lema}\label{estabilizadoresubsistem}
Let $f_1\in F$ and $f_2\in F_G$, and let $\Gamma_h$ be given by our construction. Consider the three hypercycles (two independent) that we can associate with the face $f_1$ and one independent hypercycle that we can associate with the face $f_2$. Then, there are three stabilizer generators (two independent) associated with $f_1$ and one independent stabilizer generator associated with $f_2$.
\end{lema}

\begin{proof}
Let $f_1'$ be the face inserted into $f\in F\subseteq F_R$, with the three extra edges coming from the insertion of the triangle inside $f_1'$. Define $f_{1\sigma_1}$ as the hypercycle formed by the $p_1+3$ rank-$2$ edges on the border of $f_1'$. Therefore, it follows that
\begin{equation}\label{hiperciclo1}
W(f_{1\sigma_1})=\prod_{e'\in \partial(f_1')\cap\overline{E}_R} \overline{K}_{e'} \prod_{e'\in \partial(f_1')\cap\overline{E}_G} \overline{K}_{e'}
\end{equation}
belongs to the group $\G_{loop}$, and by the Corollary \ref{coroloperadoresloops} we have that $W(f_{1\sigma_1})\in C(\G_{loop})$. Thus, $W(f_{1\sigma_1})\in\G_{loop}\cap C(\G_{loop})=\textbf{S}$.

Now, define $f_{1\sigma_2}$ to be the hypercycle consisting of all rank-$3$ edges embedded within $f_1$ and within $f_1'$, together with the rank-$2$ edges on the boundary of $ f_1$ and the rank-$2$ red edges on the border of $f_1'$. Therefore, it follows that
\begin{equation}\label{hiperciclo2}
W(f_{1\sigma_2})=\prod_{e'\in \partial(f_1)\cap\overline{E}_B} \overline{K}_{e'} \prod_{e'\in \partial(f_1)\cap\overline{E}_R} \overline{K}_{e'}\prod_{e'\in \partial(f_1')\cap\overline{E}_G} \overline{K}_{e'}
\end{equation}
belongs to the group $\G_{loop}$. We also know that every hypercycle contains none or an even number of $W(f_{1\sigma_2})$ triangles. Thus, by the Corollary \ref{coroloperadoresloops} we have that $W(f_{1\sigma_2})\in C(\G_{loop})$. Thus, $W(f_{1\sigma_2})\in\G_{loop}\cap C(\G_{loop})=\textbf{S}$.

The third and final hypercycle associated with $f_1$, which we denote by $f_{1\sigma_3}$, is the modulo $2$ sum of the hypercycles $f_{1\sigma_1}$ and $f_{1\sigma_2}$. The loop operator $W(f_{1\sigma_3})$ will be the product of $W(f_{1\sigma_1})$ by $W(f_{1\sigma_2})$.

Finally, for the face $f_2\in F_G$, define $f_{2\sigma_1}$ as the hypercycle formed by the $2p_2$ rank-$2$ edges on the border of $f_2$. Therefore, by the same reasoning used for $W(f_{1\sigma_1})$, we have that
\begin{equation}\label{hiperciclo3}
W(f_{2\sigma_1})=\prod_{e'\in \partial(f_2)\cap\overline{E}_R} \overline{K}_{e'} \prod_{e'\in \partial(f_2)\cap\overline{E}_G} \overline{K}_{e'}
\end{equation}
is an independent stabilizer generator.
\end{proof}

\begin{lema}\label{lemac2}
The stabilizer operators $W(f_{1\sigma_i})$ with $i=1,2,3$ and $W(f_{2\sigma_1})$ from Lemma \ref{estabilizadoresubsistem} satisfy the condition $(C_2)$ of the definition of topological subsystem code.
\end{lema}

\begin{proof}
It follows from the equations (\ref{hiperciclo1}), (\ref{hiperciclo2}), and (\ref{hiperciclo3}).
\end{proof}

\begin{Exam}
Let $f_1\in F\subseteq F_R$ and $f_2\in F_G$ and consider the construction of $\Gamma_h$ from tessellation $\{6,6,6\}$, as in Figure $\ref{fig2faces6-6-6}$.

\begin{figure}[!h]
	\centering
	\includegraphics[scale=1.0]{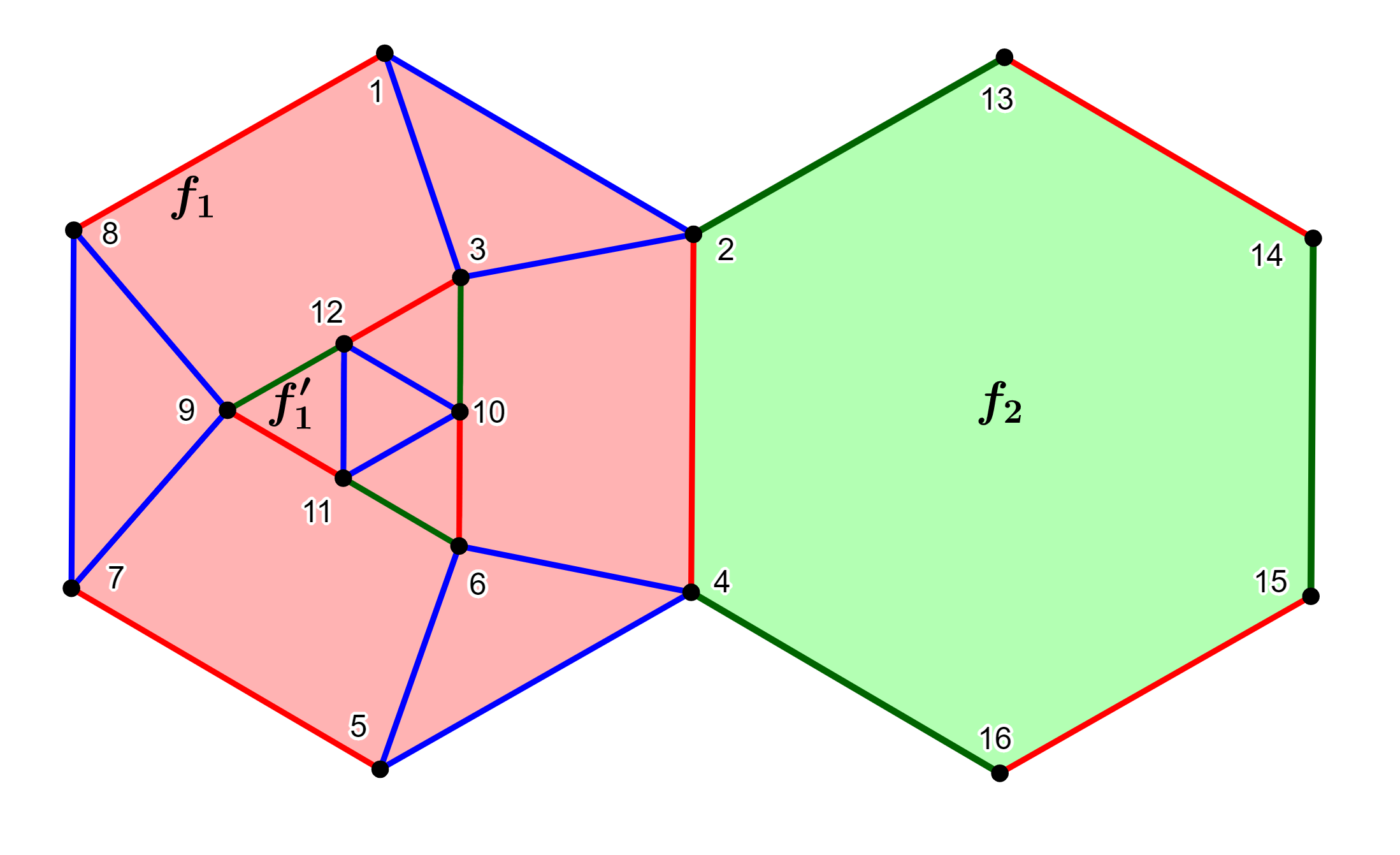} 
	\caption{Face $f_ 1\in F\subseteq F R$ and face $f_ 2\in F_G$ of the hypergraph $\Gamma_h$ obtained from tessellation $\{6,6,6\}$.} \label{fig2faces6-6-6}
\end{figure}

See that $f_{1\sigma_1}$ is given by the $6$ edges of $f_1'$ inserted inside $f_1$. Therefore, the loop operator is given as $(\ref{hiperciclo1})$, that is, $W(f_{1\sigma_1}) = (X_3X_{12})(X_{6}X_{10})(X_{9}X_{11})(Y_{3}Y_{10})(Y_{6}Y_{11})(Y_{9}Y_{12})$. Thus, $W(f_{1\sigma_1})$ satisfies the Theorem $\ref{teosindrome}$ and the condition $(C_2)$.

On the other hand, $f_{1\sigma_2}$ is formed by the three triangles inserted in $f_1$, by the triangle inserted in $f_1'$, by the red edges on the border of $f_1$, and by the red edges on the border of $f_1'$. It follows from $(\ref{hiperciclo2})$ that the loop operator is given by
\begin{eqnarray*}
 W(f_{1\sigma_2}) & = & (Z_1Z_2Z_3)(Z_4Z_5Z_6)(Z_7Z_8Z_9)(Z_{10}Z_{11}Z_{12})(X_1X_8)(X_2X_4)(X_5X_7)(X_3X_{12})\\
 & & (X_{9}X_{11})(X_{6}X_{10})\\
                 & = & (Z_1Z_2)(Z_4Z_5)(Z_7Z_8)(X_1X_8)(X_2X_4)(X_5X_7)(Y_3Y_{10})(Y_6Y_{11})(Y_9Y_{12}),  
\end{eqnarray*}
which also satisfies the Theorem $\ref{teosindrome}$ and the condition $(C_2)$.

Finally, we have that $f_{2\sigma_1}$ is formed by the $6$ rank-$2$ edges of $f_2$. Therefore, the loop operator is given as in $(\ref{hiperciclo3})$, that is, $W(f_{2\sigma_1}) = (X_2X_4)(X_{13}X_{14})(X_{15}X_{16})(Y_4Y_{16})(Y_2Y_{13})(Y_{14}Y_{15})$. Thus, $W(f_{2\sigma_1})$ satisfies the Theorem $\ref{teosindrome}$ and the condition $(C_2)$.
\end{Exam}

To conclude this section, let's remember the number of faces $N_f$, edges $N_e$, and vertices $N_v$ of the tessellation $\{2p_1,2p_2,2p_3\}$ on a compact orientable surface of genus $g\geq 2$, as well as the number of red $F_R$, green $F_G$, and blue $F_B$ faces. These values will be needed so that we can determine the parameters of three of the four families of topological subsystem codes that we will present in the next section.

As we can see in \cite{silva}, the number of faces, edges, and vertices of this tessellation is given by 
\begin{eqnarray}
N_f & = & \frac{2(p_1p_2+p_1p_3+p_2p_3)(g-1)}{p_1p_2p_3-p_1p_2-p_1p_3-p_2p_3},\label{Nf2p2q2rSub}\\ 
N_e & = & \frac{6p_1p_2p_3(g-1)}{p_1p_2p_3-p_1p_2-p_1p_3-p_2p_3},\label{Ne2p2q2rSub} \\
N_v & = & \frac{4p_1p_2p_3(g-1)}{p_1p_2p_3-p_1p_2-p_1p_3-p_2p_3} \label{Nv2p2q2rSub}.
\end{eqnarray}

The number of red, green, and blue faces is given by 
\begin{eqnarray*}
F_{R} & = & \frac{2p_2p_3(g-1)}{p_1p_2p_3-p_1p_2-p_1p_3-p_2p_3},\\  
F_{G} & = & \frac{2p_1p_3(g-1)}{p_1p_2p_3-p_1p_2-p_1p_3-p_2p_3},\\
F_{B} & = & \frac{2p_1p_2(g-1)}{p_1p_2p_3-p_1p_2-p_1p_3-p_2p_3}.\label{Nfcolors2p2q2rSub}
\end{eqnarray*}

\section{New topological subsystem codes}\label{familiassubsistem}

In this section, we present four families of subsystem codes and prove that they are, in fact, topological. We also display the $n$, $k$, $r$, and $d$ parameters of these codes.

Two of these families are particular cases of our construction, seen in the last section. The first family refers to the case in which we start from tessellation $\{2p_1,2p_2,4\}$ and consider $F = F_R$. The second and third families refer to the case in which we start from the tessellation $\{2p_1,4,6\}$ with $p_1>4$ even and with $p_1>6$ odd, respectively, and also consider $F = F_R$. As the case $p_1>4$ even was not analyzed in \cite{sarvepalli}, we will address it here. Finally, the fourth and final family that we will present is related to tessellation $\{p,4,3,4\}$ with $p$ odd, where we solve the problem of $p$ being odd by introducing a triangle inside the $p$-gon, just as we did in our construction in the last section.

These codes will be built both on the torus and on other compact orientable surfaces with genus $g\geq 2$. However, in the case of the torus, we only have one example for the first family, which is starting from the tessellation $\{6,12,4\}$, one example for the second family, which is starting from the tessellation $\{12,4 .6\}$, and no examples for the third and fourth families.

\subsection{Topological subsystem codes from tessellations $\{2p_1,2p_2,4\}$}

Consider the trivalent, $3$-colorable tessellation $\{2p_1,2p_2,4\}$ with odd $p_1>2$ on a compact orientable surface $\Msup$. Using $F = F_R$, construct the hypergraph $\Gamma_h$ as seen in Section \ref{construcaomeussubcode}. However, here, the way we choose the edges of $f'$, where we will take the vertices of the triangle that we insert inside it, cannot be done randomly.

As we saw in Section \ref{codigosubsistemtopolo}, when constructing the hypergraph $\Gamma_h$ for $p_1>2$ even and $F=F_R$, we have that each $f\in F_R$ determines three stabilizer generators (2 independent). Furthermore, because we consider $F=F_R$, we also have that each $f\in F_G$ determines three stabilizer generators (2 independent), as can be seen in Figure \ref{3hiperciclosSarvepalli}. In the proof of the Theorem \ref{teo5indiano}, which can be seen in \cite{sarvepalli}, it is shown that the latter are stabilizing generators.

\begin{figure}[!h!]
	\centering
	\includegraphics[scale=0.9]{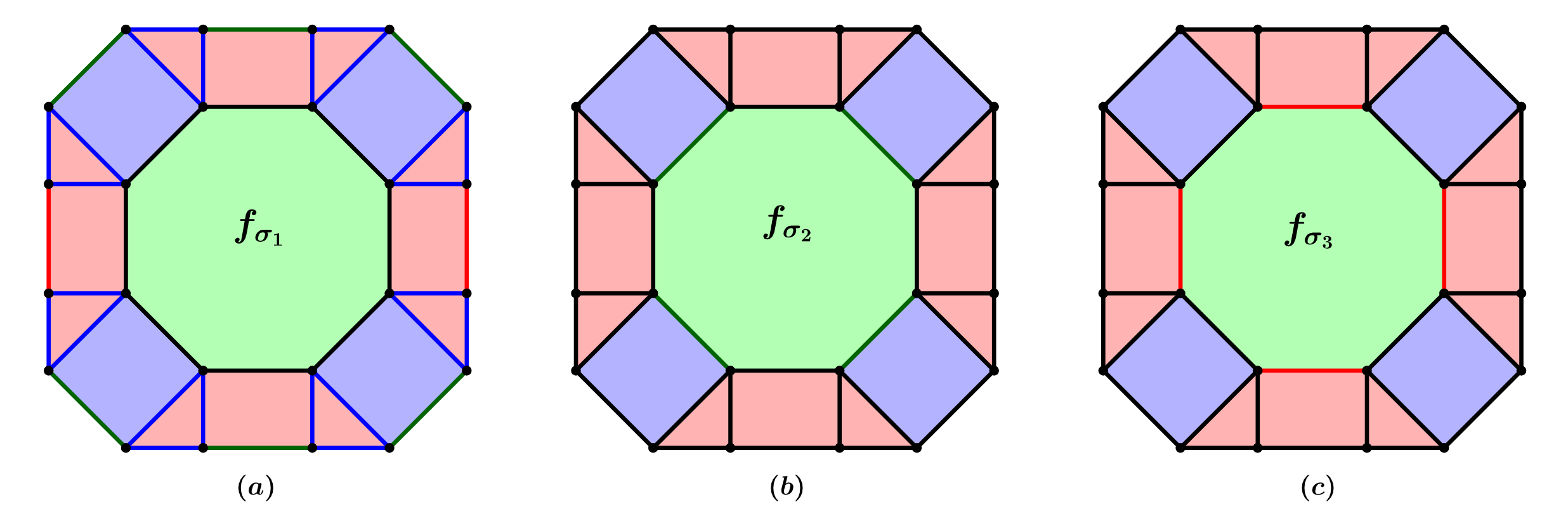} 
	\caption{Hypercycles of the face $f\in F_G$ over the hypergraph $\Gamma_h$ obtained from tessellation $\{8,8,4\}$. In black, we have the hypercycles: $(a)$ $f_{\sigma_1}$ given by the rank-$2$ edges of $f$; $(b)$ $f_{\sigma_2}$ given by the rank-$3$ edges around $f$, by the rank-2 red edges of $f$ and other red and green edges around $f$, necessary to complete the hypercycle; $(c)$ $f_{\sigma_3}$ sum modulo $2$ of $f_{\sigma_1}$ and $f_{\sigma_2}$.} \label{3hiperciclosSarvepalli}
\end{figure}

For our construction, we also want that each $f \in F_G$ provide us with three stabilizer generators ($2$ independent). But, when we add the triangles inside $f_i'$, we are creating an additional edge to the hypercycle involving the triangles around some faces $f\in F_G$. This will be a problem because, in this way, the stabilizer coming from this hypercycle will not satisfy the Theorem \ref{teosindrome}; therefore, it does not satisfy $(C_2)$, and consequently, we will not be able to obtain topological subsystem codes. To solve this, we impose the restriction that these triangles be constructed in $f_i'$, so that it does not increase new vertices or increase an even number of vertices in each hypercycle involving the triangles around $f \in F_G$, as can be seen in Figure \ref{10-8-4-hyperbolico}, where we added four new vertices to the hypercycle involving the triangles around $f$.

\begin{figure}[!h!]
	\centering
	\includegraphics[scale=7.0]{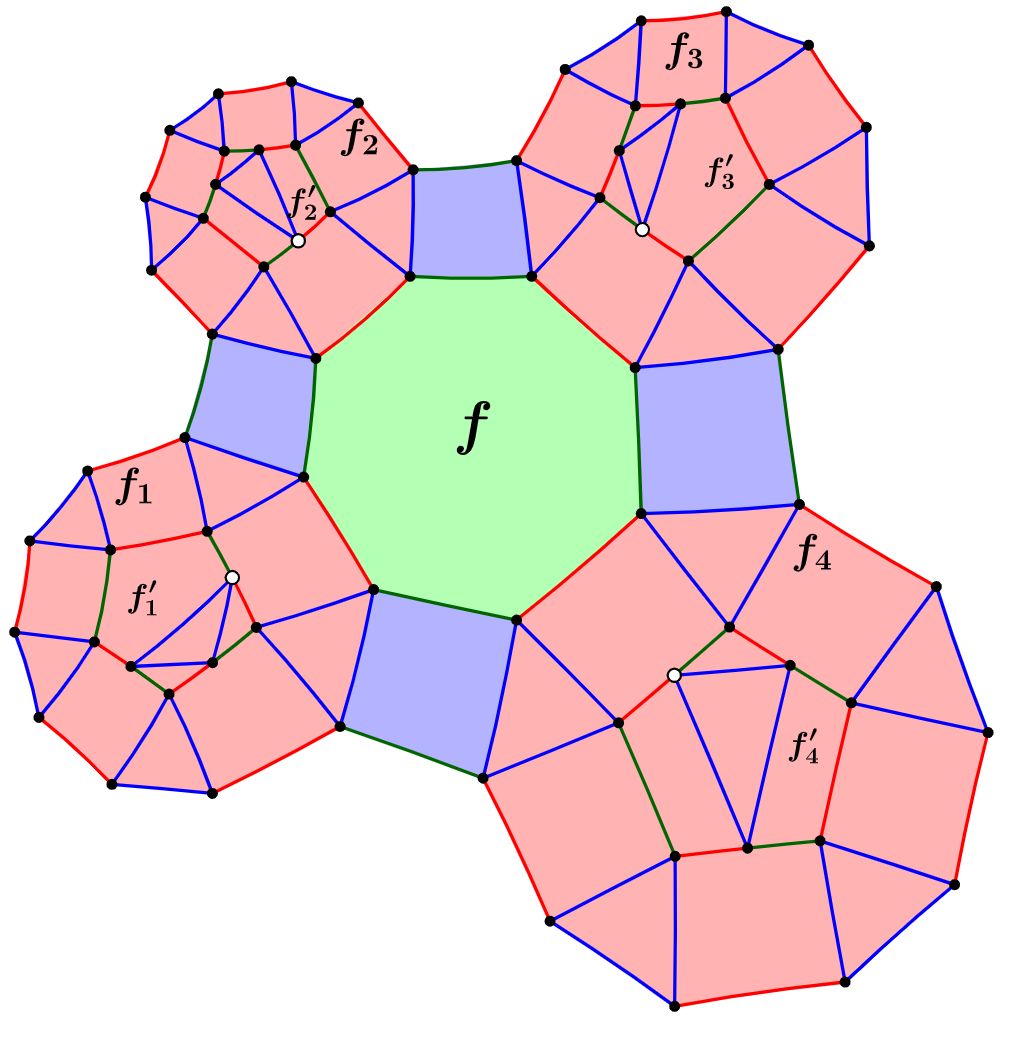} 
	\caption{Part of a hypergraph $\Gamma_h$ obtained from hyperbolic tessellation $\{10,8,4\}$, where we added $4$ new vertices (highlighted in white), that is, $4$ new edges for the hypercycle involving the triangles around $f\in F_G$.} \label{10-8-4-hiperbolico}
\end{figure}

If no vertex is increased in the hypercycle $f_{\sigma_2}$ with $f\in F_G$, we have that the loop operator will be given by
\begin{equation}\label{equa1}
W(f_{\sigma_2})=\prod_{e'\in \partial(\overline{f})\cap\overline{E}_B} \overline{K}_{e'} \prod_{e'\in \partial(\overline{f})\cap\overline{E}_R} \overline{K}_{e'} \prod_{e'\in \partial(\overline{f})\cap\overline{E}_G} \overline{K}_{e'} \prod_{e'\in \partial(f)\cap\overline{E}_R} \overline{K}_{e'},
\end{equation}
where $\overline{f}$ is the ``face" composed of the rank-$2$ edges of $\overline{\Gamma}_h$ that encompass the face $f\in F_G$ (see Figure \ref{facedefbarra1}), and are used in the hypercycles $f_{\sigma_2}$ and $f_{\sigma_3}$,

\begin{figure}[!h!]
	\centering
	\includegraphics[scale=0.9]{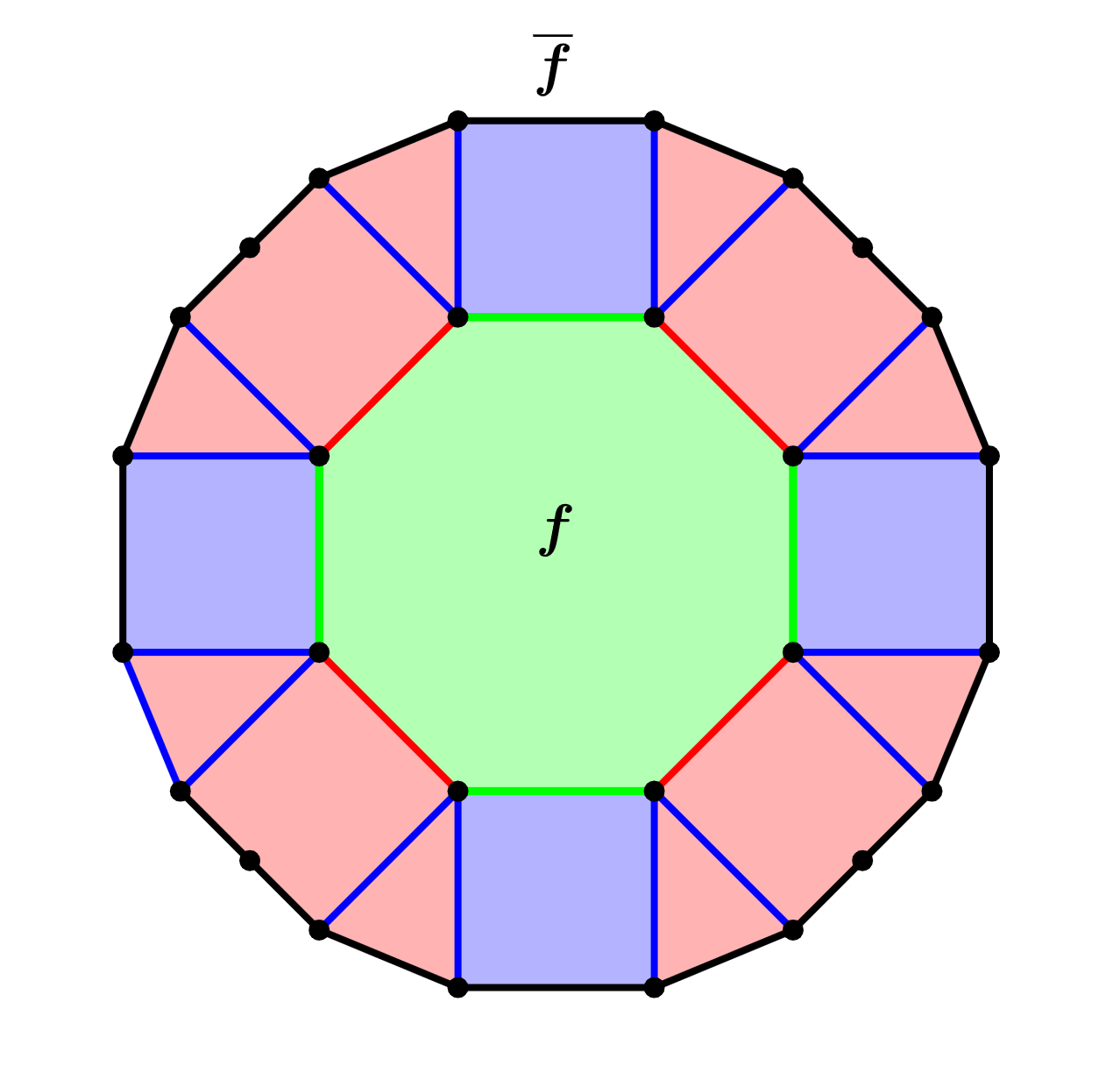} 
	\caption{Representation of the ``face" $\overline{f}$, which is composed of the black edges of $\overline{\Gamma}_h$ that encompass the face $f\in F_G$.} \label{facedefbarra1}
\end{figure}
 
If the number of vertices increases by an even amount, we will not be able to write the loop operator as in the equation (\ref{equa1}), separating by the colors of the edges, but this will not prevent the loop operator from satisfying the Theorem \ref{teosindrome}.

\begin{Exam}
Consider the face $f$ given in Figure $\ref{exemplo10-8-4}$-(a), which represents a face of $F_G$ arising from the construction of the hypergraph $\Gamma_h$ over the hyperbolic tessellation $\{10,8,4\}$. If we have no new vertices on the rank-$2$ edges of $\overline{f}$, it follows that the hypercycle $f_{\sigma_2}$ surrounding the face $f\in F_G$ gives the loop operator
\begin{eqnarray*}
W(f_{\sigma_2}) & = & (Z_1Z_2)(Z_4Z_5)(Z_7Z_8)(Z_{10}Z_{11})(Z_{13}Z_{14})(Z_{16}Z_{17})(Z_{19}Z_{20})(Z_{22}Z_{23})\\
                &   & (X_1X_{23})(X_5X_{7})(X_{11}X_{13})(X_{17}X_{19})(Y_2Y_4)(Y_8Y_{10})(Y_{14}Y_{16})(Y_{20}Y_{22})\\
                &   & (X_3X_{24})(X_6X_{9})(X_{12}X_{15})(X_{18}X_{21}).
\end{eqnarray*}
Considering that only one vertex increases (qubit $25$ in Figure $\ref{exemplo10-8-4}$-(b)), it will not be possible to write the $21$ edges in such a way that they satisfy the Theorem $\ref{teosindrome}$, because we have an odd amount of the rank-$2$ edges of $\overline{f}$, so they will never commute. Therefore, we require that the number of vertices, if it increases, be even.

On the other hand, in Figure $\ref{example10-8-4}$-(c) increased two vertices (qubits $25$ and $26$). Therefore, the loop operator will be given by
\begin{eqnarray*}
W(f_{\sigma_2}) & = & (Z_1Z_2)(Z_4Z_5)(X_7X_{26})(Y_8Y_{10})(X_{11}X_{13})(Y_{14}Y_{16})(X_{17}X_{19})(Y_{20}Y_{22})(X_{23}X_{25})\\
                &   & (Y_1Y_{25})(Y_2Y_4)(Y_5Y_{26})(Z_7Z_8)(Z_{10}Z_{11})(Z_{13}Z_{14})(Z_{16}Z_{17})(Z_{19}Z_{20})(Z_{22}Z_{23})\\
                &   & (X_3X_{24})(X_6X_{9})(X_{12}X_{15})(X_{18}X_{21}).
\end{eqnarray*}
Note that we cannot separate the link operators by colors as in the $(\ref{equa1})$ equation, but even so, this operator satisfies Theorem $\ref{teosindrome}$.
\end{Exam}

\begin{figure}[!h]
	\centering
	\includegraphics[scale=0.7]{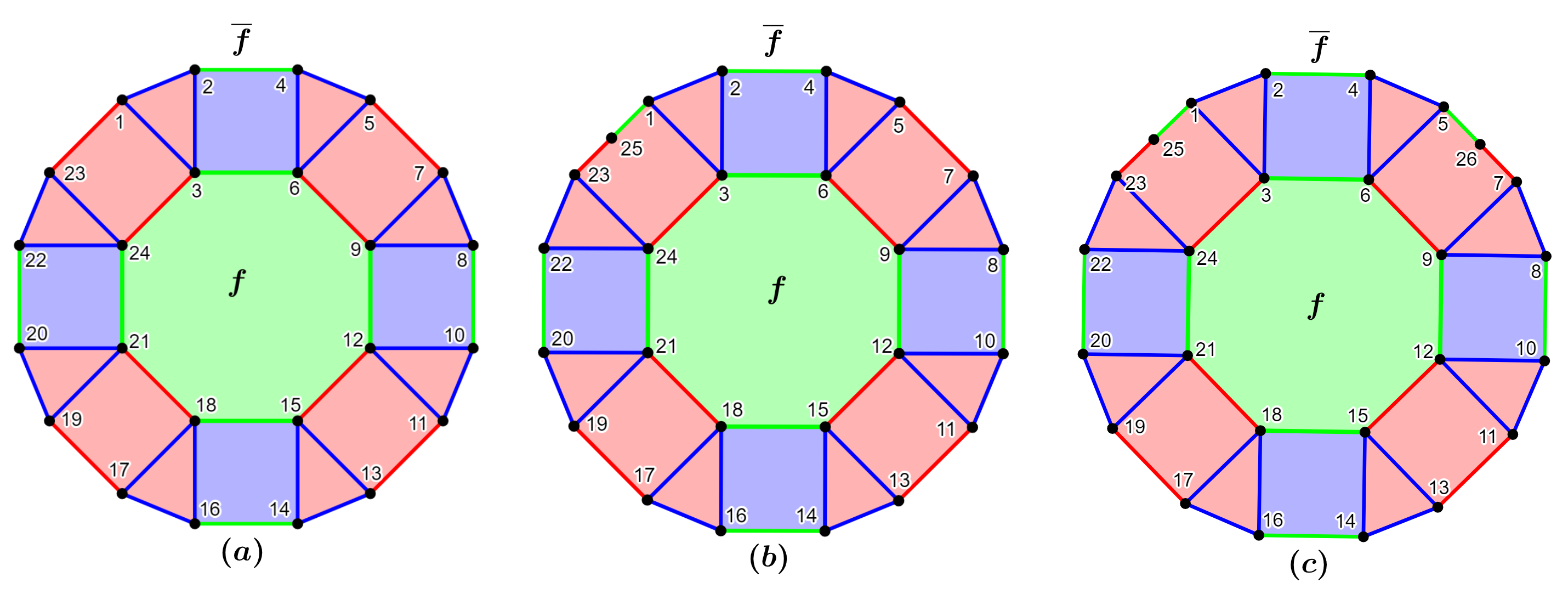} 
	\caption{Representation of a face $f\in F_G$ resulting from the construction of the hypergraph $\Gamma_h$ over tessellation $\{10,8,4\}$. $(a)$ No new vertices, $(b)$ one new vertex, and $(c)$ two new vertices on the rank-$2$ edges of $\overline{f}$.} \label{exemplo10-8-4}
\end{figure}

This guarantees that we can make this restriction regarding the way we fix the triangles inside $f_i'$ so that the hypercycles of $f\in F_G$ provide stabilizers satisfying the Theorem \ref{teosindrome}, given the fact that we have a finite amount of $2p_1$-gons and $2p_2$-gons. Therefore, we have that each $f\in F_G$ will provide three stabilizer generators ($2$ independent).

Thus, the stabilizer generators of this code will be given by emma \ref{estabilizadoresubsistem}, by the loop operators $W(f_{\sigma_2})$ and by $W(f_{\sigma_3})$ with $f\in F_G$, where $W(f_{\sigma_2})$ is the loop operator just described, and $W(f_{\sigma_3})$ is the product of $W(f_{\sigma_1})$ by $W(f_{\sigma_2})$. Therefore, this subsystem code satisfies the condition $(C_2)$.

Let us determine the number of independent stabilizer generators. It follows from what we have just seen and the Lemma \ref{estabilizadoresubsistem} that each $f\in F_R$ determines $2$ independent stabilizer generators, and each $f\in F_G$ also determines $2$ independent stabilizer generators. However, when we add all these stabilizers together, not all of them will be independent, as we have the following valid independence relationship
\begin{equation}\label{relacaosubsistem}
\prod_{f\in F_G} W(f_{\sigma_2}) = \prod_{f\in F_R} W(f_{\sigma_1}).
\end{equation}

For the case of $p_1$ even, that is, for the Theorem \ref{teo5indiano}, other independence relations may exist, as can be seen in the proof of this theorem in \cite{sarvepalli}. However, for our construction, the only valid relation in this case is the relation (\ref{relacaosubsistem}). In this way, the number of independent stabilizing generators will be
\begin{equation}\label{gerainde2p12p24}
s=2F_R+2F_G-1.
\end{equation}

Let's now determine the parameters $n$, $k$, $r$, and $d$ of our subsystem code. As we have only one Euclidean case, which is given by $\{6,12,4\}$, the determination of the parameters is made in relation to tessellations embedded in compact orientable surfaces of genus $g\geq 2$. However, the same construction applies to the case $\{6,12,4\}$, with the difference that we need to count the number of faces manually. As we know the values of $N_v$ and $F_R$, we can determine the quantity $n$ of physical qubits, which is given by the number of vertices of the hypergraph $\Gamma_h$, that is,
\begin{equation}\label{nsub}
n = N_v + p_1F_R + 3F_R = \frac{12(p_1p_2+2p_2)(g-1)}{p_1p_2-2p_1-2p_2}.
\end{equation}

We also know that the quantity of independent stabilizing generators $s$ is given in (\ref{gerainde2p12p24}) and that $\dim\G=2r+s$, so we will determine the quantity $r$ of gauge qubits. Remembering that the gauge group $\G$ is given by $\G=\langle\overline{K}_{e'};\ e'\in\overline{E}\rangle$.

To determine $\dim\G$, we need to determine the number of independent generators, that is, the number of independent link operators. As quoted in Section \ref{cap5}, the rank-$3$ edges, that is, the $\Gamma_h$ triangles, when transformed into three rank-$2$ edges, provide us with only two independent link operators. Combining these two independent link operators with all the link operators of $\Gamma_h$, we also have that they will not all be independent, as they obey the independence relations obtained for $s$, which in this case is given only by the equation (\ref {relacaosubsistem}). Thus, we have that the number of independent generators of $\G$ will be the number of rank-$2$ edges of $\Gamma_h$, plus twice the number of rank-$3$ edges, minus the number of independence relations. We have that the rank-$2$ and rank-$3$ edges for the hypergraph $\Gamma_h$, obtained by our construction on the tessellation $\{2p_1,2p_2,4\}$, are given by
\begin{equation}
E_2 = 2p_2F_G + (p_1 + 3 )F_R\ \ \textrm{e}\ \ E_3 = (p_1 + 1)F_R. 
\end{equation}
Thus,
\begin{eqnarray*}
\dim\G & = & 2p_2F_G + (p_1+3)F_R + 2(p_1+1)F_R-1\\
       & = & 2p_2F_G + 3p_1F_R+5F_R - 1\\
       & = & N_v+N_e+5F_r-1,
\end{eqnarray*}
since $2p_2F_G = N_v$ e $3p_1F_R = N_e$. As $\dim\G=2r+s$ e $s=2F_G+2F_R-1$, we have
\begin{eqnarray*}
2r+2F_G+2F_R-1 & = & N_v+N_e+5F_R-1 \Leftrightarrow\\
2r & = & N_v+N_e+5F_R-2F_G-2F_R-2F_B+2F_B+2N_v-2N_v\\
   & = & -2(N_v-N_e+N_f)+2N_v+5F_R\\
   & = & -2\chi +2N_v+5F_R,
\end{eqnarray*}
since $F_G+F_R+F_B = N_f$ e $2F_B+N_v=N_e$. Therefore, 
\begin{equation}\label{rsub}
r = -\chi +N_v+\frac{5F_R}{2},
\end{equation}
which will only make sense if $r$ is a positive integer. As we already know the values of $n$, $r$, and $s$, it is easy to determine the value of $k$, as it is given by $k=n-r-s$. So,
\begin{eqnarray*}
k & = & N_v+p_1F_R+3F_R+\chi - N_v-\frac{5F_R}{2}-2F_R-2F_G+1\\
  & = & N_e-N_v-2F_R-2F_G-2F_B+2F_B+3F_R+\chi -\frac{5F_R}{2}+1\\
  & = & \frac{-2N_f+4F_B+6F_R-5F_R+2+2N_v-2N_v+2N_e-2N_e}{2}\\
  & = & \frac{-2(N_v-N_e+N_f)+3N_v-2N_e+F_R+2}{2}\\
  & = & -\chi + 1 + \frac{F_R}{2},
\end{eqnarray*}
since $N_v+p_1F_R = N_e$, $4F_B=N_v$ e $3N_v-2N_e=0$. Thus, 
\begin{equation}\label{ksub}
k=-\chi + 1 + \frac{F_r}{2},
\end{equation}
which will only make sense if $k$ is a positive integer. Finally, the distance $d$ follows in a similar way to what was done in Section \ref{cap5}, that is, we will limit it by the smallest number of triangles in a homological non-trivial closed hypercycle, that is, it will be the smallest weight among the operators of $C(\G)\setminus\textbf{S}$.

The next two results are used to prove that $\Gamma_h$ satisfies the conditions $(C_1)$ and $(C_3)$ of the topological subsystem code definition. These results can be seen in \cite{sarvepalli}, but here we will prove that they are valid for our construction.

\begin{lema}\label{lemasub1}
Let $\sigma$ be a homologically non-trivial closed hypercycle of $\Gamma_h$. Then $\sigma$ must contain edges of rank-$3$.
\end{lema}

\begin{proof}
Suppose the hypercycle $\sigma$ does not contain rank-$3$ edges. It follows from our construction that every vertex of the hypergraph $\Gamma_h$ is trivalent and has a rank-$3$ edge incident on it. We also have that $\Gamma_h$ satisfies $(H_5)$, and the rank-$3$ edges are all blue. Thus, $\sigma$ is formed by rank-$2$ edges that have alternating colors between red and green.

If $\sigma$ contains some vertex $v \in f_1'$ for $f_1'$ inserted in $f_1 \in F = F_R$, then all $p_1 + 3$ vertices of $f_1'$ are in $\sigma$, that is, the hypercycle $f_{1\sigma_1}$, which is homologically trivial, is part of $\sigma$. We can make the sum modulo $2$ of $\sigma$ with $f_{1\sigma_1}$, thus discarding $f_{1\sigma_1}$ of $\sigma$. This way, only hypercycles that do not contain vertices of $f_1'$ will remain in $\sigma$. We can do this for all faces $f'$ that have vertices in $\sigma$. Therefore, the remaining vertices of $\sigma$ will be vertices of faces of $F_G$. Again, due to the fact that $\sigma$ does not have rank-$3$ edges, it follows that it will have all $2p_2$ vertices of some face $f_2\in F_G$, that is, it has $f_{2\sigma_1}$. Then, $\sigma$ is a homologically trivial hypercycle, which is absurd. Therefore, $\sigma$ must contain rank-$3$ edges.
\end{proof}

\begin{lema}\label{lemasub2}
Let $\sigma$ be a homologically non-trivial closed hypercycle of $\Gamma_h$. Then, $W(\sigma)$ is not in the gauge group $\G$.
\end{lema}

\begin{proof}
Suppose $W(\sigma)\in \G$. Then 
\[
W(\sigma) = \prod_{e'\in E_2\cap\sigma} K_{e'}\prod_{e'\in E_3\cap\sigma} K_{e'}\in\G,
\]
where $E_2$ and $E_3$ are the sets of rank-$2$ and rank-$3$ edges of $\Gamma_h$, respectively. We know that the edges of $E_2\cap\sigma$ are edges of $\overline{\Gamma}_h$, and the link operators $K_{e'}$ and $\overline{K}_{e'}$ are the same as $e'\in E_2\cap\sigma$. Thus, taking the product of $W(\sigma)\in\G$ and $\prod_{e'\in E_2\cap\sigma} \overline{K}_{e'}\in\G$, it follows that the operator of the type $Z$ 
\[
O_\sigma = \prod_{e'\in E_3\cap\sigma} K_{e'} = W(\sigma)\prod_{e'\in E_2\cap\sigma} \overline{K}_{e'}\in\G,
\]
that is, it will be generated by the link operators of the form $\{X\otimes X,Y\otimes Y, Z\otimes Z\}$.

Suppose $O_{\sigma}$ is generated by $O_{\sigma} = K^{(X,Y)}K^{(Z)}$, where $K^{(X,Z)}$ consists only of operators of the form $X\otimes X$ and $Y\otimes Y$. $K^{(Z)}$ only consists of operators of the form $Z\otimes Z$. Therefore, $O_{\sigma}K^{(Z)}=K^{(X,Y)}$ and the edges in the support of $K^{(X,Y)}$ must form a closed hypercycle consisting only of rank-$2$ edges. According to the previous lemma, this closed hypercycle is homologically trivial and consists of the union of hypercycles that are boundaries of a collection of faces $f\in F_G$ and $f_{i}'$ with $f_i\in F_R$. We denote this collection by $\F$. Thus, the rank-$3$ edges of $\sigma$ are incident on the vertices of this hypercycle obtained through $K^{(X,Y)}$.

Consider the operator $O'=\prod_{f\in\F} W(f_{\sigma_i})$, where $f_{\sigma_i}$ is a hypercycle associated with $f$, which contains all edges of rank-$3$ of $\sigma$.

\noindent Claim: The rank-$3$ edges, which are incident on the vertices in the support of $K^{(X,Y)}$ but are not in $\sigma$, are also not in $O'$.

\noindent Indeed, suppose that $v$ is not a vertex in the support of $O_{\sigma}$ but is a vertex in the support of $K^{(X,Y)}$. Then, the three gauge operators $X_vX_i$, $Y_vY_j$, and $Z_vZ_u$ act on $v$. Therefore, suppose, without loss of generality, that $u$ is a vertex in the support of $K^{(Z)}$. Thus, $u$ is not a vertex in the support of $O_{\sigma}$, otherwise, as $u$ belongs to the support of $K^{(Z)}$, then $v$ would belong to the support of $O_{\sigma}$, which is absurd.

Now, let $e'\in E_3\setminus\sigma$ be such that $e'$ has exactly two vertices in the support of $O_{\sigma}K^{(Z)}=K^{(X,Y)}$. Therefore, there are two faces $f_a$ and $f_b$ associated with these two vertices. Thus, there is a hypercycle $f_{a\sigma_2}$ that encompasses $f_a$ and contains $e'$, and there is a hypercycle $f_{b\sigma_2}$ that encompasses $f_b$ and contains $e'$. So the operator $W(f_{a\sigma_2})W(f_{b\sigma_2})$ is not supported in $e'$. Therefore, the vertices of $e'$ are not in the support of $O'$.

Continuing this procedure, we obtain stabilizers that are supported on the same set of edges as $O_{\sigma}$, which implies that $\sigma$ is homologically trivial, as it will be generated by a combination of trivial hypercycles, which is absurd. Therefore, $W(\sigma)\not\in\G$.
\end{proof}

\begin{Rema}
Lemma $\ref{lemasub1}$ guarantees that the homologically non-trivial closed hypercycles $\sigma$ have edges of rank-$3$. This is very important, because otherwise $W(\sigma)$ would be a stabilizer and $(C_1)$ and $(C_3)$ would not be satisfied. The Lemma $\ref{lemasub2}$ guarantees that $W(\sigma)\not\in\G =C(\G_{loop})$, thus guaranteeing that logical operators can be identified with homologically closed hypercycles that are not trivial, and stabilizer operators can be identified with homologically trivial closed hypercycles, that is, $(C_1)$ and $(C_3)$ are satisfied.
\end{Rema}

\begin{Theo}\label{T1meusub}
Consider a trivalent, $3$-colorable tessellation $\{2p_1,2p_2,4\}$ with odd $p_1>2$. Applying the construction of $\Gamma_h$ with $F=F_R$ seen in the section $\ref{construcaomeussubcode}$, we obtain topological subsystem codes with parameters
\begin{eqnarray}\label{parasub}
[[N_e+3F_R,-\chi +1+\frac{F_R}{2},-\chi + N_v+\frac{5F_R}{2},d\leq l]]
\end{eqnarray}
where $l$ is the number of triangles in a homologically non-trivial hypercycle.
\end{Theo}

\begin{proof}
It follows from Theorem \ref{satisfazasrestricoes} that $\Gamma_h$ provides subsystem codes, and it follows from Lemma \ref{lemac2} and the equation (\ref{equa1}) that this subsystem code satisfies $(C_2)$. By the Lemmas \ref{lemasub1} and \ref{lemasub2}, we have that it will satisfy $(C_1)$ and $(C_3)$. Therefore, the code obtained is a topological subsystem code. Finally, by the equations (\ref{nsub}), (\ref{rsub}), and (\ref{ksub}) we have that the parameters $n$, $r$, and $k$ are given as in (\ref{parasub}). The fact that the distance is limited by the number of triangles in a homologically non-trivial hypercycle follows from the discussion in this section.
\end{proof}

\begin{figure}[!h]
	\centering
	\includegraphics[scale=0.9]{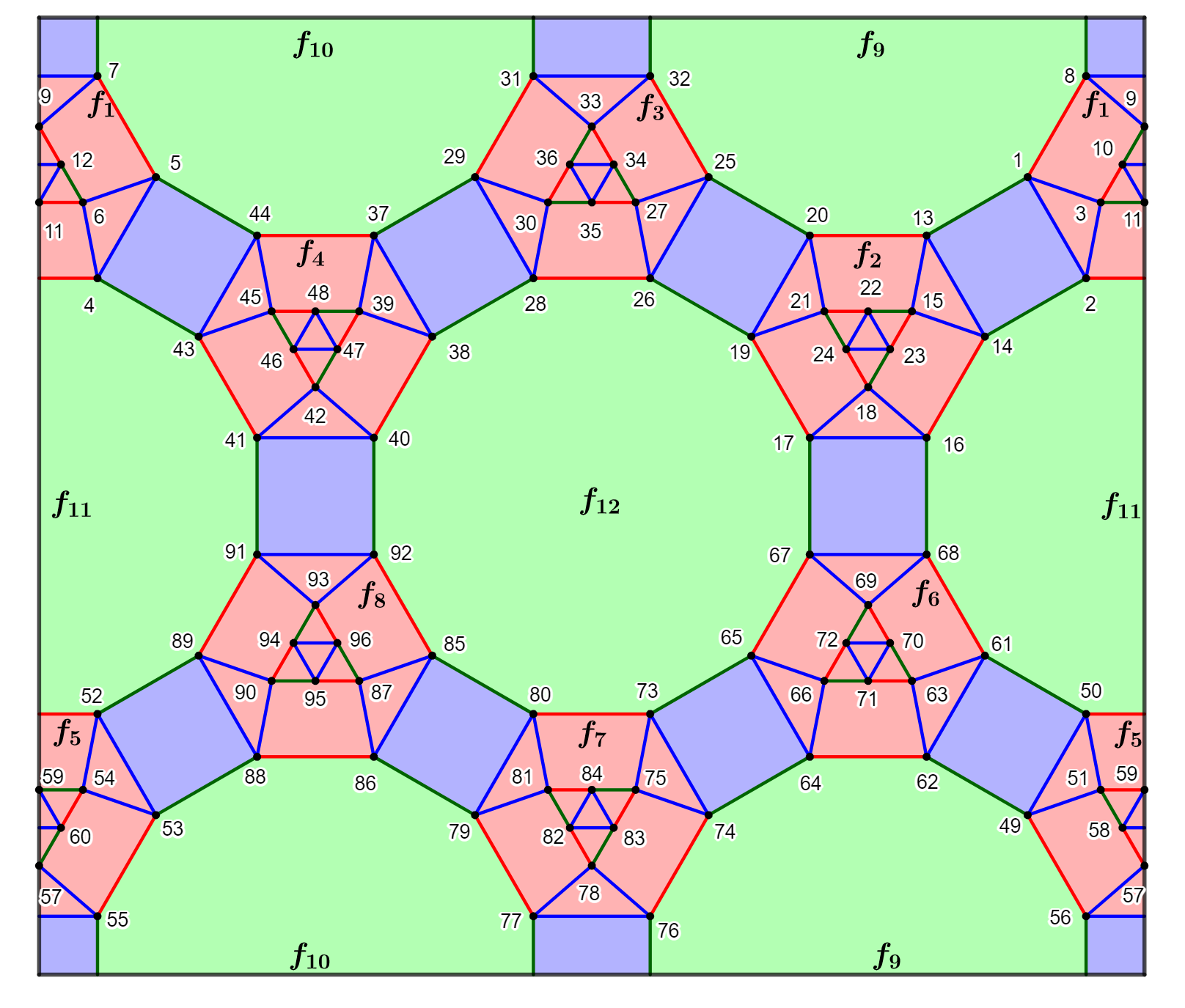} 
	\caption{Hypergraph $\Gamma_h$ arising from tessellation $\{6,12,4\}$.} \label{exemplo96qubit}
\end{figure}

\begin{Exam}
Consider the hypergraph $\Gamma_h$ arising from tessellation $\{6,12,4\}$ given in Figure $\ref{exemplo96qubit}$. By the Theorem $\ref{T1meusub}$, we obtain a topological subsystem code through $\Gamma_h$. Let's determine the parameters of this code.

This hypergraph has $96$ vertices, that is, $96$ qubits. Therefore, it follows that $n=96$. To determine $r$, we need to determine $s$, $E_2$, $E_3$, and $\dim\G$. Remembering that the number of independent stabilizer generators is given by $s = 2F_R+2F_G-1$ and, in this case, $F_R=8$ and $F_G=4$, therefore, $s=23$. $E_2$ and $E_3$ are given by $E_2=2p_2F_G+(p_1+3)F_R$ and $E_3=(p_1+1)F_R$. Then, in this case, $E_2=96$ and $E_3=32$. Since $\dim\G = E_2+2E_3-1$, it follows that $\dim\G=159$. But $\dim\G=2r+s$, so it follows that $159=2r+23\Longleftrightarrow r=68$. Now, since $k=n-r-s$, it follows that $k=5$.

Finally, searching for homologically non-trivial hypercycles, the one with the smallest number of triangles found has $l=4$ triangles. Therefore, $d\leq 4$. Therefore, we obtain a topological subsystem code with parameters $[[96,5,68,4]]$.
\end{Exam}

To conclude the study on this family of codes, we present in Tables \ref{t1subsistem} and \ref{t2subsistem} the parameters $n$, $k$ and $r$, obtained from the hypergraph $\Gamma_h$ originating from the tessellations $\{2p_1,2p_2,4\}$ on a compact orientable surface $\Msup$ with genus $g=2, 3, 4$ and $5$. We also present in each case the quantity $s$ of independent stabilizer generators together with the tessellations $\{2p_1,2p_2,4\}$ used to construct $\Gamma_h$. We do not present the value of $d$ in any of the tables, because despite having an upper bound for the distance, we still cannot determine a value in general.

\begin{table}[h]
\centering
\begin{multicols}{2}
\begin{tabular}{|w{c}{2mm}|w{c}{18mm}|w{c}{3mm}|w{c}{28mm}|}
\hline	$g$	&	Tesselação	&	$s$	&	$[[n,k,r,d]]$\\
\hline	2	& $\{6,14,4\}$	& 79 &	$[[336,17,240,d]]$\\
\hline	2	& $\{14,6,4\}$	& 79 &  $[[288,9,200,d]]$\\
\hline	2	& $\{6,16,4\}$	& 43 &	$[[192,11,138,d]]$\\
\hline	2	& $\{6,18,4\}$	& 31 &	$[[144,9,104,d]]$\\
\hline	2	& $\{10,8,4\}$	& 35 &  $[[144,7,102,d]]$\\
\hline	2	& $\{6,20,4\}$	& 25 &	$[[120,8,87,d]]$\\
\hline	2	& $\{18,6,4\}$	& 31 &  $[[120,5,84,d]]$\\
\hline	2	& $\{6,24,4\}$	& 19 &	$[[96,7,70,d]]$\\
\hline	2	& $\{6,36,4\}$	& 13 &	$[[72,6,53,d]]$\\
\hline	2	& $\{10,10,4\}$	& 15 &	$[[72,5,52,d]]$\\
\hline	2	& $\{10,20,4\}$	& 5	 &	$[[36,4,27,d]]$\\
\hline	3	& $\{6,14,4\}$	& 159 &	$[[672,33,480,d]]$\\
\hline	3	& $\{14,6,4\}$	& 159 &	$[[576,17,400,d]]$\\
\hline	3	& $\{6,16,4\}$	& 87  &	$[[384,21,276,d]]$\\
\hline	3	& $\{6,18,4\}$	& 63  &	$[[288,17,208,d]]$\\
\hline 
\end{tabular}
\vspace{5mm}

\begin{tabular}{|w{c}{2mm}|w{c}{18mm}|w{c}{3mm}|w{c}{28mm}|}
\hline	$g$	&	Tesselação	& $s$ &	$[[n,k,r,d]]$\\
\hline	3	& $\{10,8,4\}$	& 71  &	$[[288,13,204,d]]$\\
\hline	3	& $\{6,20,4\}$	& 51  &	$[[240,15,174,d]]$\\
\hline	3	& $\{18,6,4\}$	& 63  &	$[[240,9,168,d]]$\\
\hline	3	& $\{6,24,4\}$	& 39  &	$[[192,13,140,d]]$\\
\hline	3	& $\{6,28,4\}$	& 33  &	$[[168,12,123,d]]$\\
\hline	3	& $\{6,36,4\}$	& 27  &	$[[144,11,106,d]]$\\
\hline	3	& $\{10,10,4\}$	& 31  &	$[[144,9,104,d]]$\\
\hline	3	& $\{6,60,4\}$	& 21  &	$[[120,10,89,d]]$\\
\hline	3	& $\{10,12,4\}$	& 21  &	$[[108,8,79,d]]$\\
\hline	3	& $\{10,20,4\}$	& 11  &	$[[72,7,54,d]]$\\
\hline	3	& $\{18,12,4\}$	& 9	  &	$[[60,6,45,d]]$\\
\hline	3	& $\{14,28,4\}$	& 5   &	$[[48,6,37,d]]$\\
\hline
\end{tabular}
\end{multicols}
\caption{Parameters of topological subsystem codes from tessellations $\{2p_1,2p_2,4\}$ with $p_1>2$ odd, $g=2$ and $3$.}\label{t1subsistem}
\end{table}

\begin{table}[!h]
\centering
\begin{multicols}{2}
\begin{tabular}{|w{c}{2mm}|w{c}{18mm}|w{c}{3mm}|w{c}{28mm}|}
\hline	$g$	& Tesselação   & $s$ & $[[n,k,r,d]]$\\
\hline	4	& $\{6,14,4\}$ & 239 & $[[1008,49,720,d]]$\\
\hline	4	& $\{14,6,4\}$ & 239 & $[[864,25,600,d]]$\\
\hline	4	& $\{6,16,4\}$ & 131 & $[[576,31,414,d]]$\\
\hline	4	& $\{6,18,4\}$ & 95	 & $[[432,25,312,d]]$\\
\hline	4	& $\{10,8,4\}$ & 107 & $[[432,19,306,d]]$\\
\hline	4	& $\{6,20,4\}$ & 77	 & $[[360,22,261,d]]$\\
\hline	4	& $\{18,6,4\}$ & 95	 & $[[360,13,252,d]]$\\
\hline	4	& $\{6,24,4\}$ & 59	 & $[[288,19,210,d]]$\\
\hline	4	& $\{6,30,4\}$ & 47	 & $[[240,17,176,d]]$\\
\hline	4	& $\{6,36,4\}$ & 41	 & $[[216,16,159,d]]$\\
\hline	4	& $\{10,10,4\}$& 47	 & $[[216,13,156,d]]$\\
\hline	4	& $\{6,48,4\}$ & 35	 & $[[192,15,142,d]]$\\
\hline	4	& $\{14,8,4\}$ & 43  & $[[192,11,138,d]]$\\
\hline	4	& $\{30,6,4\}$ & 47	 & $[[192,9,136,d]]$\\
\hline	4	& $\{6,84,4\}$ & 29	 & $[[168,14,125,d]]$\\
\hline	4	& $\{10,20,4\}$& 17	 & $[[108,10,81,d]]$\\
\hline	4	& $\{14,14,4\}$& 15	 & $[[96,9,72,d]]$\\
\hline	4	& $\{18,36,4\}$& 5	 & $[[60,8,47,d]]$\\
\hline	5	& $\{6,14,4\}$ & 319 & $[[1344,65,960,d]]$\\
\hline	5	& $\{14,6,4\}$ & 319 & $[[1152,33,800,d]]$\\
\hline
\end{tabular}
\vspace{5mm }

\begin{tabular}{|w{c}{2mm}|w{c}{18mm}|w{c}{3mm}|w{c}{28mm}|}
\hline	$g$	& Tesselação   & $s$ & $[[n,k,r,d]]$\\
\hline	5	& $\{6,16,4\}$ & 175 & $[[768,41,552,d]]$\\
\hline	5	& $\{6,18,4\}$ & 127 & $[[576,33,416,d]]$\\\hline	5	& $\{10,8,4\}$ & 143 & $[[576,25,408,d]]$\\\hline	5	& $\{6,20,4\}$ & 103 & $[[480,29,348,d]]$\\
\hline	5	& $\{18,6,4\}$ & 127 & $[[480,17,336,d]]$\\
\hline	5	& $\{6,24,4\}$ & 79	 & $[[384,25,280,d]]$\\
\hline	5	& $\{6,28,4\}$ & 67	 & $[[336,23,246,d]]$\\
\hline	5	& $\{6,36,4\}$ & 55	 & $[[288,21,212,d]]$\\
\hline	5	& $\{10,10,4\}$& 63	 & $[[288,17,208,d]]$\\
\hline	5	& $\{6,44,4\}$ & 49	 & $[[264,20,195,d]]$\\
\hline	5	& $\{6,60,4\}$ & 43	 & $[[240,19,178,d]]$\\
\hline	5	& $\{6,108,4\}$& 37	 & $[[216,18,161,d]]$\\
\hline	5	& $\{10,12,4\}$& 43	 & $[[216,15,158,d]]$\\
\hline	5	& $\{10,20,4\}$& 23	 & $[[144,13,108,d]]$\\
\hline	5	& $\{14,12,4\}$& 25	 & $[[144,12,107,d]]$\\
\hline	5	& $\{18,12,4\}$& 19	 & $[[120,11,90,d]]$\\
\hline	5	& $\{10,60,4\}$& 13	 & $[[108,12,83,d]]$\\
\hline	5	& $\{14,28,4\}$& 11	 & $[[96,11,74,d]]$\\
\hline	5	& $\{30,12,4\}$& 13	 & $[[96,10,73,d]]$\\
\hline	5	& $\{22,44,4\}$& 5	 & $[[72,10,57,d]]$\\
\hline 
\end{tabular}
\end{multicols}
\caption{Parameters of topological subsystem codes from tesellations $\{2p_1,2p_2,4\}$ with $p_1>2$ odd, $g=4$ and $5$.}\label{t2subsistem}
\end{table}

\subsection{Topological subsystem codes From tessellation $\{2p_1,4,6\}$}

Consider the trivalent and $3$-colorable tessellation $\{2p_1,4,6\}$ on a compact orientable surface $\Msup$. We will construct two families of topological subsystem codes. The first one is for even $p_1>4$. This family is not considered in \cite{sarvepalli}. In this way, we will provide a family of codes following the construction given in \cite{sarvepalli}, as we saw in Section \ref{codigosubsistemtopolo}. The second family will be for odd $p_1>6$, where we make our construction of $\Gamma_h$ given in Section \ref{construcaomeussubcode}. For both cases, we keep fixing the $2p_1$-gons as red faces, the $4$-gons as green faces, and the $6$-gons as blue faces.

\subsection*{Codes From tessellation $\{2p_1,4,6\}$ with even $p_1>4$}

Using $F=F_R$, construct the hypergraph $\Gamma_h$ as seen in Section \ref{codigosubsistemtopolo}, so that the triangles border the $6$-gons as in Figure \ref{hipergrafo12-4-6--1}. Note that the only Euclidean example of this case is $\{12,4,6\}$.

\begin{figure}[!h]
	\centering
	\includegraphics[scale=0.15]{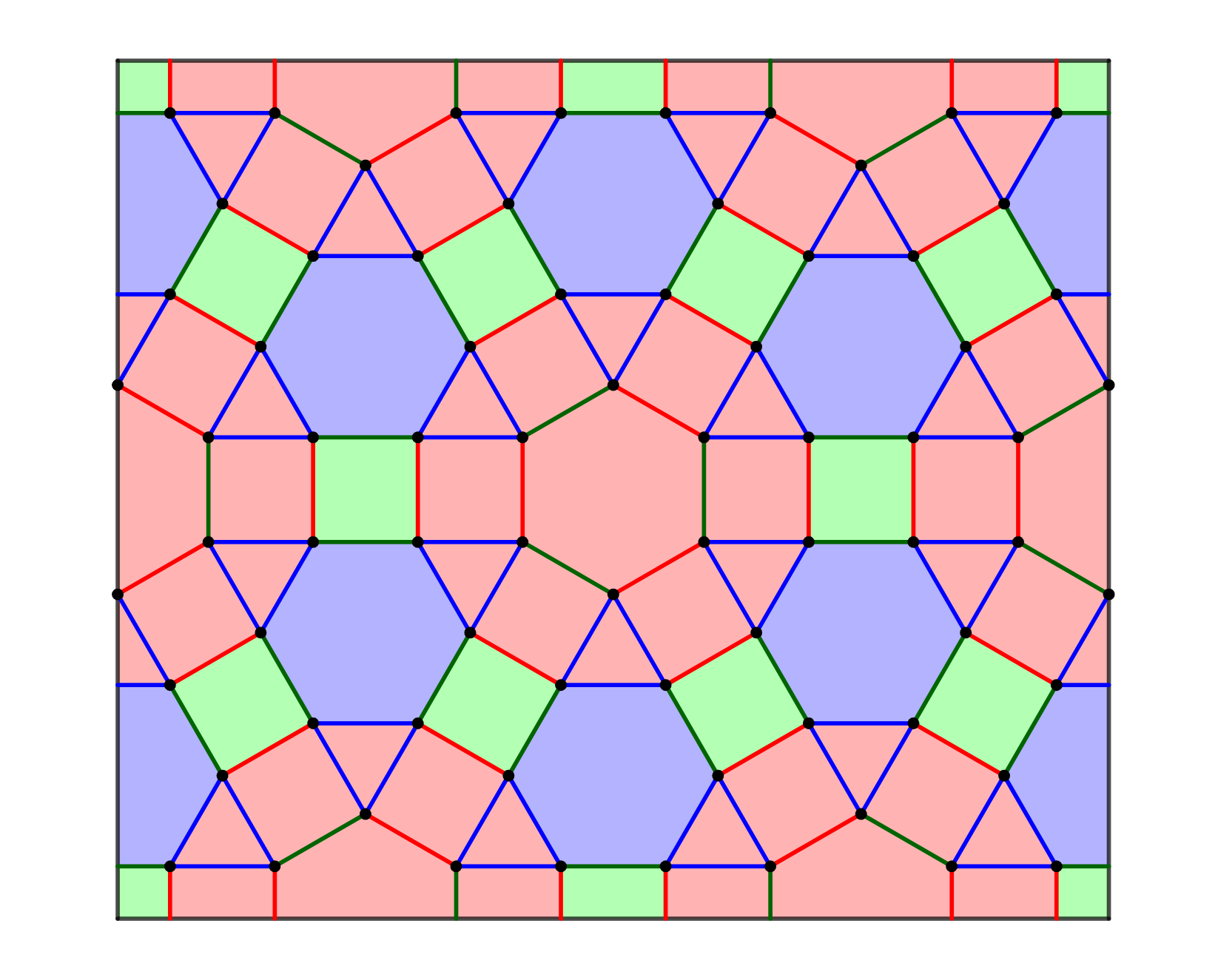} 
	\caption{Hypergraph $\Gamma_h$ arising from tessellation $\{12,4,6\}$.} \label{hipergrafo12-4-6--1}
\end{figure}

As in Section \ref{codigosubsistemtopolo}, each $f\in F_R$ provides $3$ stabilizer generators ($2$ independent), and, just as in the Lemma \ref{estabilizadoresubsistem}, we will have a stabilizer for each face $f\in F_G$, which is given by

\begin{equation}\label{equa2}
W(f_{\sigma_1}) = \prod_{e'\in\partial(f)\cap\overline{E}_R} \overline{K}_{e'}\prod_{e'\in\partial(f)\cap\overline{E}_G} \overline{K}_{e'}.
\end{equation}
However, in this case, even with $F=F_R$, there is no other stabilizer for $f\in F_G$. But there will be a third independent stabilizer generator for each face $f\in F_R$, which comes from a hypercycle, denoted by $\overline{f}_\sigma$, as in Figure \ref{hipergrafo12-4-6---2}-($a$), where $\overline{f}$ is the ``face" given in Figure \ref{hipergrafo12-4-6---2}-($b$). See that the stabilizer coming from the hypercycle $\overline{f}_\sigma$ satisfies the Theorem \ref{teosindrome}, since it can be written as
\begin{equation} \label{equa3}
\begin{aligned}
W(\overline{f}_\sigma) = \prod_{e'\in\partial(\overline{f})\cap\overline{E}_B} \overline{K}_{e'} \prod_{e'\in\partial(\overline{f})\cap\overline{E}_G} \overline{K}_{e'} \prod_{e'\in\partial(\overline{f})\cap\overline{E}_R} \overline{K}_{e'} \prod_{\substack{e'\in\partial(f_i)\cap\overline{E}_G \\ f_i\in F_G;\ f_i\subset \overline{f}}} \overline{K}_{e'}\prod_{e'\in\partial(f)\cap\overline{E}_B} \overline{K}_{e'}\prod_{e'\in\partial(f)\cap\overline{E}_R} \overline{K}_{e'}
\end{aligned}
\end{equation}

\begin{figure}[!h]
	\centering
	\includegraphics[scale=0.145]{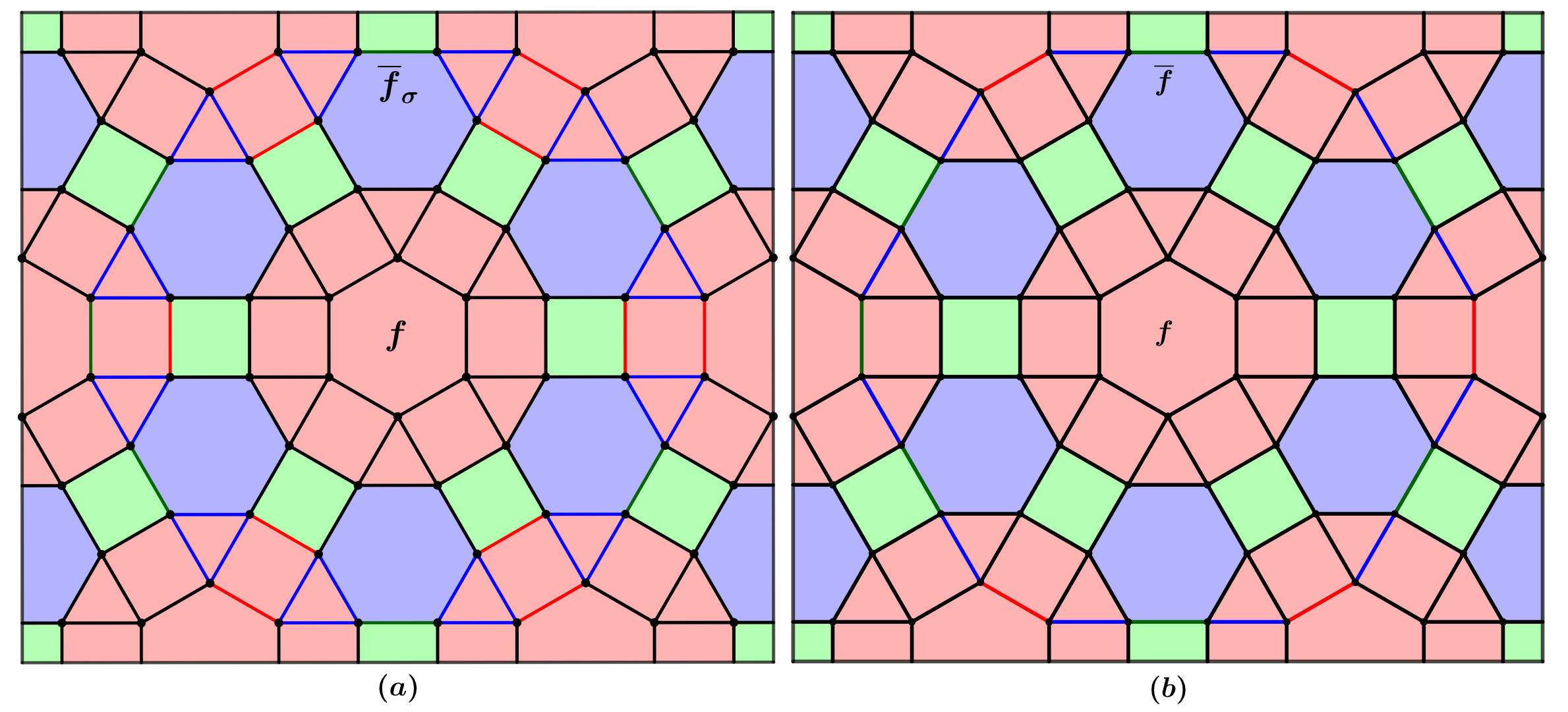} 
	\caption{($a$) Hypercycle $\overline{f}_\sigma$ with colored edges, which corresponds to face $f\in F_R$. ($b$) Colored edges highlighting the ``face" $\overline{f}$.} \label{hipergrafo12-4-6---2}
\end{figure}

%\begin{eqnarray}
%W(\overline{f}_\sigma) & = & (Z_{1}Z_{2})(Z_{3}Z_{4})(Z_{5}Z_{6})(Z_{7}Z_{8})(Z_{9}Z_{10})(Z_{11}Z_{12})(Z_{13}Z_{14})(Z_{15}Z_{16})(Z_{17}Z_{18})\nonumber\\
%                       &   & (Z_{19}Z_{20})(Z_{21}Z_{22})%(Z_{23}Z_{24})(X_{1}X_{24})(Y_{2}Y_{3})(X_{4}X_{5})(Y_{6}Y_{7})(X_{8}X_{9})(Y_{10}Y_{11})\nonumber\\
%                       &   & (Y_{12}Y_{13})(Y_{14}Y_{15})(X_{16}X_{17})(Y_{18}Y_{19})(X_{20}X_{21})(Y_{22}Y_{23})(Y_{25}Y_{37})(Y_{26}Y_{38})\nonumber\\
%                       &   & (Y_{9}Y_{39})(Y_{28}Y_{40})(Y_{29}Y_{41})(Y_{30}Y_{42})(Y_{31}Y_{43})(Y_{32}Y_{44})(Y_{33}Y_{45})(Y_{34}Y_{46})(Y_{35}Y_{47})\nonumber\\
%                       &   & (Y_{36}Y_{48})(Z_{37}Z_{48})(Z_{38}Z_{39})(Z_{40}Z_{41})(Z_{42}Z_{43})(Z_{44}Z_{45})(Z_{46}Z_{47})(X_{37}X_{38})\nonumber\\
%                       &   & (X_{39}X_{40})(X_{41}X_{42})%(X_{43}X_{44})(X_{45}X_{46})(X_{47}X_{48}).
%\end{eqnarray}

Thus, the stabilizer generators will be given as we saw in Section \ref{codigosubsistemtopolo}, that is, they are given by $W(f_{\sigma_1})$, $W(f_{\sigma_2})$, and $W(f_{\sigma_3})$, where $W(f_{\sigma_3}) = W(f_{\sigma_1})W(f_{\sigma_2})$. They are also given by (\ref{equa2}) and (\ref{equa3}). Therefore, this subsystem code satisfies the condition $(C_2)$.

We need to determine the number of independent stabilizer generators. We have that each $f\in F_R$ determines three independent stabilizer generators, and each $f\in F_G$ determines one single independent stabilizer generator. However, when we add all these stabilizers together, not all of them will be independent, as we have the following valid independence relationship
\begin{equation}\label{RE1} 
\prod_{f\in F_R} W(\overline{f}_{\sigma}) = \prod_{f\in F_R} W(f_{\sigma_1})\prod_{f\in F_G} W(f_{\sigma_1}).
\end{equation}

Now, let $\Gamma_R$ be the red reduced tessellation obtained from $\{2p_1,4,6\}$, that is, we reduce the red faces to a point and connect these points, so that we obtain triangles as in Figure \ref{tesselacaoreduzida}. If $\Gamma_R$ is tripartite, that is, if the set of vertices admits a partition into three sets, such that every edge has its end points in different sets, then we divide the faces of $F_R$ into three sets, according to the tripartition, which we denote by $F_1$, $F_2$, and $F_3$. Define $F^{F_i}_G$ to be the set of green faces that are adjacent to the hypercycles $\overline{f}_{\sigma}$ with $f\in F_i$ and $i=1,2,3$. Therefore, there are three independence relations
\begin{eqnarray}
\prod_{f\in F_1} W(\overline{f}_{\sigma}) \prod_{f\in F^{F_1}_G} W(f_{\sigma_1}) & = & \prod_{f\in F_2} W(f_{\sigma_2})\prod_{f\in F_3} W(f_{\sigma_2}),\\
\prod_{f\in F_2} W(\overline{f}_{\sigma}) \prod_{f\in F^{F_2}_G} W(f_{\sigma_1}) & = & \prod_{f\in F_1} W(f_{\sigma_2})\prod_{f\in F_3} W(f_{\sigma_2}),\\
\prod_{f\in F_3} W(\overline{f}_{\sigma}) \prod_{f\in F^{F_3}_G} W(f_{\sigma_1}) & = & \prod_{f\in F_1} W(f_{\sigma_2})\prod_{f\in F_2} W(f_{\sigma_2}),
\end{eqnarray}
of which only two are independent. Together with (\ref{RE1}), we will have three independence relations, and, therefore, the number of independent stabilizer generators will be

\begin{equation}\label{estabisub1}
s=3F_R + F_G -1-2\delta_{\Gamma_R},
\end{equation}
onde $\delta_{\Gamma_R}=1$ se $\Gamma_R$ for tripartida e 0, caso contrário.
\begin{figure}[!h]
	\centering
	\includegraphics[scale=0.15]{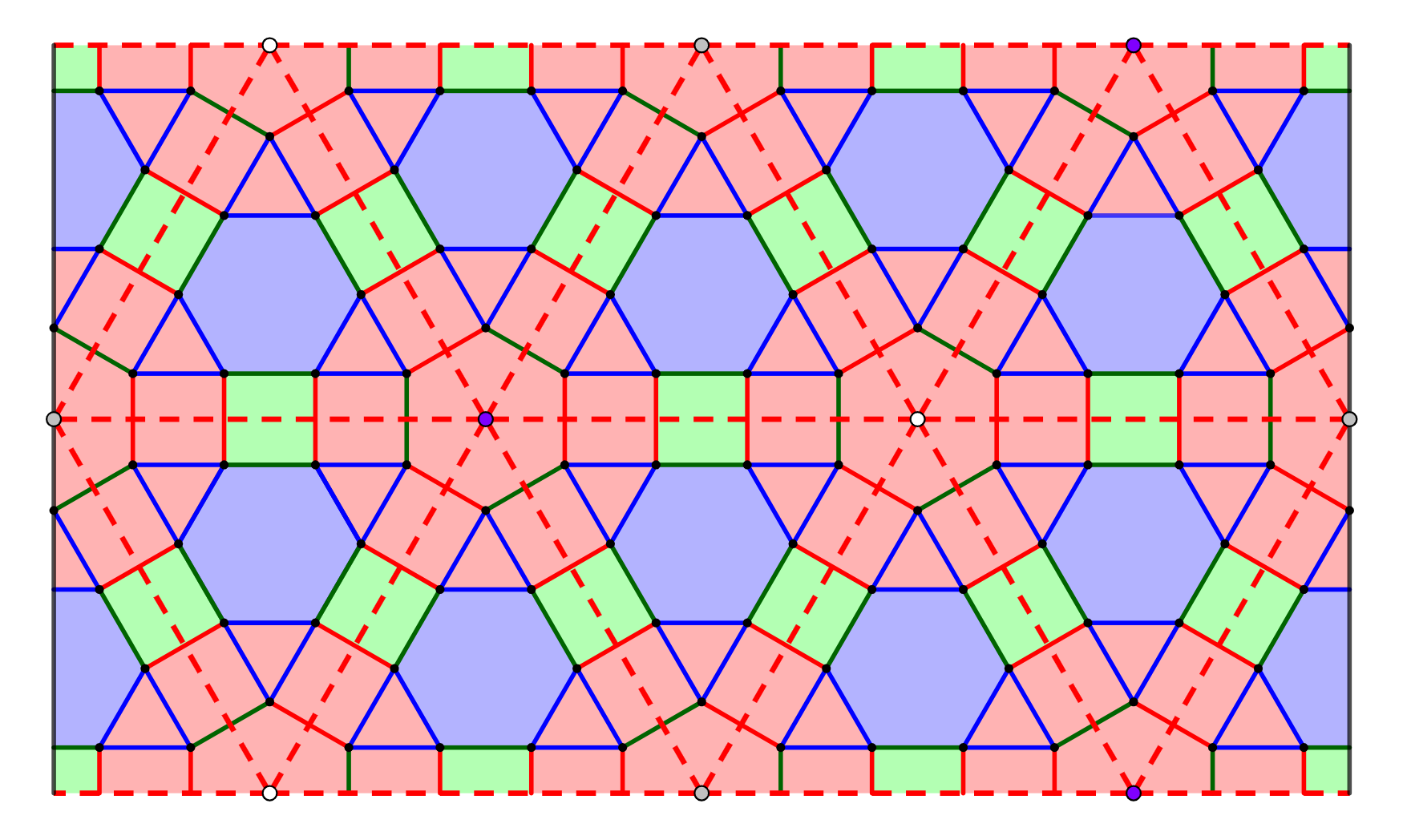} 
	\caption{Red dashed edges correspond to reduced tessellation $\Gamma_R$ of the hypergraph $\Gamma_h$. See that $\Gamma_R$ is tripartite.} \label{tesselacaoreduzida}
\end{figure}

Let's now determine the parameters $n$, $k$, $r$, and $d$ of our subsystem code. As we have only one Euclidean case, $\{12,4,6\}$, the determination of the parameters is made in relation to the hyperbolic tessellations embedded in compact orientable surfaces of genus $g\geq 2$. However, the same construction applies to the case $\{12,4,6\}$, with the difference that we will have to count the number of faces manually. As we know the values of $N_v$ and $F_R$, we can determine the quantity $n$ of physical qubits, which is given by the number of vertices of the hypergraph $\Gamma_h$, that is,
\begin{equation}
n = N_v+p_1F_R = \frac{36p_1(g-1)}{p_1-6}=N_e.
\end{equation}

We also know that the number of independent stabilizer generators $s$ is given in (\ref{estabisub1}) and that $\dim\G = 2r + s$. But, to determine $\dim\G$, we need to determine the number of independent generators of $\G$, that is, we need to determine the number of independent link operators. We have that the number of rank-2 and rank-3 edges for the hypergraph $\Gamma_h$ is given by
\begin{equation}
E_2 = 4F_G+p_1 F_R\ \ \textrm{e}\ \ E_3=p_1F_R.
\end{equation}
Thus, 
\begin{eqnarray*}
\dim\G & = & 4F_G+p_1 F_R + 2p_1F_R -1-2\delta_{\Gamma_R}\\
       & = & N_v+N_e-1-2\delta_{\Gamma_R},
\end{eqnarray*}
since $4F_G=N_v$ and $3p_1F_R=N_e$. As $\dim\G2r+s$ and $s=3F R + F G -1-2\delta_{\Gamma_R}$, making the appropriate substitutions, we determine that the quantity $r$ of gauge qubits is given by
\begin{equation}
r = -\chi + 2N_v-\frac{N_f+N_e}{2},
\end{equation}
where $r$ must be a positive integer. 

As we already know the values of $n$, $r$, and $s$ it is easy to determine the value of $k$, as it is given by $k = n - r - s$. Thus, making the appropriate substitutions, we have that the quantity $k$ of logical qubits is given by
\begin{equation}
k = \chi -3F_R + \frac{F}{2} + 1 +2\delta_{\Gamma_R},
\end{equation}
where $k$ must be a positive integer. 

Finally, the distance $d$ will be the same as for the family constructed previously, that is, we get an upper bound to the code distance by the smallest number of triangles in a homologically non-trivial closed hypercycle, that is, it will be the smallest weight among the $C(\G)\setminus\textbf{S}$ operators.

Finally, the distance $d$ is calculated in the same way as in the previous case, thus we obtain an upper limit for the code distance considering the smallest number of triangles in a non-trivial homologically closed hypercycle, that is, it will be the smallest weight among the $C(\G)\setminus\textbf{S}$ operators.

The Lemmas \ref{lemasub1} and \ref{lemasub2} are also valid here, and their proofs are done in a similar way. From this, it follows that this code satisfies the conditions $(C_1)$ and $(C_3)$. We have already seen that this code satisfies $(C_2)$ and it was shown in \cite{sarvepalli} that the graph $\Gamma_h$ built on this type of tessellation satisfies $(H_1),\ldots,(H_5)$. Therefore, this demonstrates the following theorem.

\begin{Theo}\label{T2meusub}
Consider a trivalent and $3$-colorable tessellation $\{2p_1,4,6\}$ with even $p_1>4$. Applying the construction of $\Gamma_h$ given in \cite{sarvepalli} with $F=F_R$ such that the edges of rank-$3$ are the border of the $6$-gons, we obtain topological subsystem codes with parameters
\begin{equation}
[[N_e, \chi -3F_R + \frac{F}{2} + 1 +2\delta_{\Gamma_R},  -\chi + 2N_v-\frac{N_f+N_e}{2}, d\leq l]], 
\end{equation}
where $l$ is the number of triangles in a homologically non-trivial hypercycle.
\end{Theo}

\begin{Exam}
Consider the hypergraph $\Gamma_h$ arising from tessellation $\{6,12,4\}$ illustrated in Figure $\ref{hipergrafo12-4-6--1}$. By the Theorem $\ref{T2meusub}$, we obtain a topological subsystem code through $\Gamma_h$. Let's determine the parameters of this code.

This hypergraph has $72$ vertices, that is, $72$ qubits. Then, $n=72$. To determine $r$, we need to determine $s$, $E_2$, $E_3$, and $\dim\G$. Remembering that the number of independent stabilizer generators is given by $s=3F_R + F_G -1-2\delta_{\Gamma_R}$ and, in this case, $F_R=4$ and $F_G=12$. As the reduced tessellation $\Gamma_R$ is not tripartite, it follows that $\delta_{\Gamma_R}=0$, therefore, $s=23$. $E_2$ and $E_3$ are given by $E_2 = 4F_G+p_1 F_R$ and $E_3=p_1F_R$. Therefore, in this case, $E_2=72$ and $E_3=24$. Since $\dim\G = E_2+2E_3-1$, it follows that $\dim\G=119$. But $\dim\G=2r+s$, so it follows that $119=2r+23\Longleftrightarrow r=48$. Now, since $k=n-r-s$, we have that $k=1$.

Finally, searching for homologically non-trivial hypercycles, the one with the smallest number of triangles found has $l=4$ triangles. Therefore, $d\leq 4$. Therefore, we obtain a topological subsystem code with parameters $[[72,1,48,4]]$.
\end{Exam}

To conclude the study on this family of codes, we present in the table \ref{t5subsistem} the parameters obtained from the hypergraph $\Gamma_h$ arising from the tessellations $\{2p_1,4,6\}$ on a compact orientable surface $\Msup$ with genders $g=2$, $3$, $4$, and $5$. As we cannot determine whether or not the graph will be tripartite, we will determine the two possible values for the encoded qubits, which we call $k_1$ and $k_2$, respectively. We also present in each case the tessellations $\{2p_1,4,6\}$ used to construct $\Gamma_h$, together with the quantity $s_1$ and $s_2$ of independent stabilizer generators, if the graph is tripartite and if not, respectively.

\begin{table}[h!]
\centering
\begin{tabular}{|w{c}{2mm}|w{c}{18mm}|w{c}{3mm}|w{c}{28mm}|w{c}{3mm}|w{c}{28mm}| }
\hline 	$g$	& Tessellation & $s_1$ & $[[n,k_1,r,d]]$ & $s_2$	& $[[n,k_2,r,d]]$\\
\hline 2 & $\{16,4,6\}$ & 39 & $[[144,6,99,d]]$ & 41 & $[[144,4,99,d]]$\\
\hline	2	&	$\{20,4,6\}$	&	21	&	$[[90,6,63,d]]$ &	23	&	$[[90,4,63,d]]$\\
\hline	2	&	$\{24,4,6\}$	&	15	&	$[[72,6,51,d]]$	&	17	&	$[[72,4,51,d]]$\\
\hline	2	&	$\{36,4,6\}$	&	9	&	$[[54,6,39,d]]$	&	11	&	$[[54,4,39,d]]$\\
\hline	3	&	$\{16,4,6\}$	&	81	&	$[[288,9,198,d]]$	&	83	&	$[[288,7,198,d]]$\\
\hline	3	&	$\{20,4,6\}$	&	45	&	$[[180,9,126,d]]$	&	47	&	$[[180,7,126,d]]$\\
\hline	3	&	$\{24,4,6\}$	&	33	&	$[[144,9,102,d]]$	&	35	&	$[[144,7,102,d]]$\\
\hline	3	&	$\{28,4,6\}$	&	27	&	$[[126,9,90,d]]$	&	29	&	$[[126,7,90,d]]$\\
\hline	3	&	$\{36,4,6\}$	&	21	&	$[[108,9,78,d]]$	&	23	&	$[[108,7,78,d]]$\\
\hline	3	&	$\{60,4,6\}$	&	15	&	$[[90,9,66,d]]$	&	17	&	$[[90,7,66,d]]$\\
\hline	4	&	$\{16,4,6\}$	&	123	&	$[[432,12,297,d]]$	&	125	&	$[[432,10,297,d]]$\\
\hline	4	&	$\{20,4,6\}$	&	69	&	$[[270,12,189,d]]$	&	71	&	$[[270,10,189,d]]$\\
\hline	4	&	$\{24,4,6\}$	&	51	&	$[[216,12,153,d]]$	&	53	&	$[[216,10,153,d]]$\\
\hline	4	&	$\{36,4,6\}$	&	33	&	$[[162,12,117,d]]$	&	35	&	$[[162,10,117,d]]$\\
\hline	4	&	$\{48,4,6\}$	&	27	&	$[[144,12,105,d]]$	&	29	&	$[[144,10,105,d]]$\\
\hline	4	&	$\{84,4,6\}$	&	21	&	$[[126,12,93,d]]$	&	23	&	$[[126,10,93,d]]$\\
\hline	5	&	$\{16,4,6\}$	&	165	&	$[[576,15,396,d]]$	&	167	&	$[[576,13,396,d]]$\\
\hline	5	&	$\{20,4,6\}$	&	93	&	$[[360,15,252,d]]$	&	95	&	$[[360,13,252,d]]$\\
\hline	5	&	$\{24,4,6\}$	&	69	&	$[[288,15,204,d]]$	&	71	&	$[[288,13,204,d]]$\\
\hline	5	&	$\{28,4,6\}$	&	57	&	$[[252,15,180,d]]$	&	59	&	$[[252,13,180,d]]$\\
\hline	5	&	$\{36,4,6\}$	&	45	&	$[[216,15,156,d]]$	&	47	&	$[[216,13,156,d]]$\\
\hline	5	&	$\{44,4,6\}$	&	39	&	$[[198,15,144,d]]$	&	41	&	$[[198,13,144,d]]$\\
\hline	5	&	$\{60,4,6\}$	&	33	&	$[[180,15,132,d]]$	&	35	&	$[[180,13,132,d]]$\\
\hline
\end{tabular}
\caption{Parameters of topological subsystem codes from tessellation $\{2p_1,4,6\}$ with even $p_1>4$, $g=2$, $3$, $4$, and $5$, where $s_1$ and $ k_1$ are the values when $\Gamma_R$ is tripartite and $s_2$ and $k_2$ are the values when $\Gamma_R$ is not tripartite.}\label{t5subsistem}
\end{table}

\subsection*{Codes from tessellation $\{2p_1,4,6\}$ with odd $p_1>6$}

Using $F=F_R$, construct the hypergraph $\Gamma_h$ as in Section \ref{construcaomeussubcode}. Lemma \ref{estabilizadoresubsistem} gives that each $f\in F_R$ provides $3$ stabilizer generators ($2$ independent) and each $f\in F_G$ provides $1$ independent stabilizer generator. We want the stabilizer generators $W(\overline{f}_\sigma)$ with $f\in F_R$, seen for the case $p_1>4$ even, to continue being stabilizers. But when we add the triangles inside $f'$, we are creating an extra vertex in some of the $\overline{f}_\sigma$ hypercycles, just as it happened for the $\overline{f}_{\sigma_2} hypercycles $ with $f\in F_G$ from the first family of codes we built. This will be a problem, because, in this way, the stabilizer coming from the hypercycle $\overline{f}_\sigma$ will not satisfy the Theorem \ref{teosindrome}, therefore, it does not satisfy $(C_2)$, and consequently, we will not obtain topological subsystem codes.

Again, to solve this, we impose the restriction that these triangles be constructed in $f'$, so that it does not increase new vertices or increase an even number of vertices in each hypercycle of $\overline{f}_\sigma$.

If we have no new vertices in $\overline{f}_\sigma$ hypercycle, we have that the loop operator will be given by (\ref{equa3}). If an even number of vertices increases, we will not be able to write the loop operator as in the equation (\ref{equa3}), separating by the colors of the edges, but this will not prevent the loop operator from satisfying the \ref{teosindrome} Theorem.

Again, which guarantees that we can make this restriction on the way we fix the triangles inside $f'$, so that the hypercycles of $\overline{f}_\sigma$ provide stabilizers satisfying the Theorem \ref{teosindrome}, is the fact that we have a finite amount of $2p_1$-gons. Therefore, it follows that each $f\in F_R$ will provide us with four stabilizing generators ($3$ independent), and each $f\in F_G$ will provide an independent stabilizer generator. Thus, the stabilizer generators of this code will be given by the Lemma \ref{estabilizadoresubsistem} and the loop operators $W(\overline{f}_{\sigma})$. Therefore, this subsystem code satisfies the condition $(C_2)$.

Let us determine the number of independent stabilizer generators. For odd $p_1>6$, we have that $\Gamma_R$ will never be tripartite. Thus, the only independence relation that will be valid is (\ref{RE1}). Therefore, the number of independent stabilizing generators will be
\begin{equation}
s = 3F_R + F_G -1.
\end{equation}

Now we determine the parameters $n$, $k$, $r$, and $d$ of our subsystem code. Note that there are no Euclidean tessellations for this case. Therefore, the determination of the parameters will be made in relation to the tessellations embedded in compact orientable surfaces of genus $g\geq 2$. The quantity of physical qubits, $n$, is given by the quantity of vertices of the hypergraph $\Gamma_h$,
\begin{equation}
n = N_v+p_1F_R+3F_R = N_e + 3F_R.
\end{equation}

The quantity of gauge qubits, $r$, is determined in a similar way to the other cases. Here, the rank-$2$ and rank-$3$ edges are given by
\begin{equation}
E_2 = 4F_G +p_1F_R+3F_R\ \ \textrm{e}\ \ E_3 = p_1F_R+ F_R.
\end{equation}
Then,
\begin{eqnarray*}
\dim\G & = & 4F_G +p_1F_R+3F_R +2p_1F_R+ 2F_R-1 \\
       & = & N_v+N_e+5F_R-1.
\end{eqnarray*}
As $\dim\G = 2r+s$ and $s=3F_R+F_G-1$, we have
\begin{equation}
r = -\chi+2N_v-\frac{N_f+N_e}{2}+\frac{5F_R}{2},
\end{equation}
where $r$ is a positive integer. Now, for the parameter $k$, we know that $k=n-r-s$. Then, we have
\begin{equation}
k = \chi + \frac{N_f}{2}-\frac{5F_R}{2}+1,
\end{equation}
where $k$ is a positive integer.

Finally, the distance $d$ follows as in the other cases studied, that is, it will be upper bound by the smallest number of triangles in a homologically non-trivial closed hypercycle, that is, it will be the smallest weight among the operators of $C(\G)\setminus\textbf{S}$.

Lemmas \ref{lemasub1} and \ref{lemasub2} are also valid in this case, and their proofs are done in a similar way. From this, it follows that this code satisfies the conditions $(C_1)$ and $(C_3)$. We have already seen that this code satisfies $(C_2)$, and it is given in the Theorem \ref{satisfazasrestricoes} that the graph $\Gamma_h$ satisfies $(H_1),\ldots,(H_5)$. Therefore, this proves the following theorem.

\begin{Theo}
Consider a trivalent, $3$-colorable tessellation $\{2p_1,4,6\}$ with odd $p_1>6$. Applying the construction of $\Gamma_h$ with $F=F_R$ seen in Section $\ref{construcaomeussubcode}$, we obtain topological subsystem codes with parameters
\begin{equation}
[[N_e+3F_R, \chi + \frac{N_f}{2}-\frac{5F_R}{2}+1,  -\chi+2N_v-\frac{N_f+N_e}{2}+\frac{5F_R}{2}, d\leq l]], 
\end{equation}
where $l$ is the number of triangles in a homologically non-trivial hypercycle.
\end{Theo}

To conclude, we present in Table \ref{t7subsistem} the parameters $n$, $k$, and $r$, obtained from the hypergraph $\Gamma_h$ originating from the tessellations $\{2p_1,4,6\}$ on a compact orientable surface $\Msup$ with genus $g=2, 3, 4$ and $5$. We also present in each case the quantity $s$ of independent stabilizer generators together with the tessellations $\{2p_1,4,6\}$ used to construct $\Gamma_h$.

\begin{table}[h!]
\centering
\begin{tabular}{|w{c}{2mm}|w{c}{18mm}|w{c}{3mm}|w{c}{28mm}|w{c}{3mm}|w{c}{28mm}| }
\hline	$g$	&	Tessellation	&	$s$	&	$[[n,k,r,d]]$\\
\hline	2	&	$\{14,4,6\}$	&	77	&	$[[288,10,201,d]]$\\
\hline	2	&	$\{18,4,6\}$	&	29	&	$[[120,6,85,d]]$\\
\hline	3	&	$\{14,4,6\}$	&	155	&	$[[576,19,402,d]]$\\
\hline	3	&	$\{18,4,6\}$	&	59	&	$[[240,11,170,d]]$\\
\hline	4	&	$\{14,4,6\}$	&	233	&	$[[864,28,603,d]]$\\
\hline	4	&	$\{18,4,6\}$	&	89	&	$[[360,16,255,d]]$\\
\hline	4	&	$\{30,4,6\}$	&	41	&	$[[192,12,139,d]]$\\
\hline	5	&	$\{14,4,6\}$	&	311	&	$[[1152,37,804,d]]$\\
\hline	5	&	$\{18,4,6\}$	&	119	&	$[[480,21,340,d]]$\\
\hline
\end{tabular}
\caption{Parameters of topological subsystem codes coming from tessellations $\{2p_1,4,6\}$ with odd $p_1>6$ for $g=2, 3, 4$ and $5$.}\label{t7subsistem}
\end{table}

\subsection*{Codes from tessellation $\{p,4,3,4\}$}

Here, we construct topological subsystem codes based on the tessellation $\{p,4,3,4\}$ for odd $p\geq 7$, which is obtained from the tessellation $\{p,3\}$ as follows: Given a tessellation $\{p,3\}$ on a compact surface $\Msup$ of genus $g\geq 2$, take the dual tessellation $\{3,p\}$ and keep the two tessellations $\{p,3\}$ and $\{3,p\}$ overlapping. In this way, we form quadrilaterals with angles $\frac{2\pi}{p}$, $\frac{\pi}{2}$, $\frac{2\pi}{3}$, and $\frac{\pi }{2}$. As each quadrilateral has an incenter, we take the incenter of each of these quadrilaterals. Connecting these incenters, we obtain the tessellation $\{p,4,3,4\}$, as can be seen in the figure \ref{bombin2}-($a$).

We saw the construction of the topological subsystem codes from \cite{bombinsub} in Section \ref{codigosubsistemtopolo}: it begins from a trivalent and $3$-colorable tessellation $\{2p_1,2p_2,2p_3\}$. When this tessellation is regular, that is, when it is of the form $\{p,3\}$, after carrying out the constructions of triangles and squares, the tessellation $\{p,4,3,4\}$ is obtained with even $p$; that is, it is the same as carrying out the process described above, starting from the tessellation $\{p,3\}$ with even $p$.

If $p$ is odd, the tessellation $\{p,3\}$ is not $3$-colorable, so the construction given in \cite{bombinsub} does not apply in this case. That's why we are considering an alternative approach to working with the odd $p$ case.

The process of building this family of topological subsystem codes is very simple and similar to what we did in Section \ref{construcaomeussubcode}. First, we take a tessellation $\{p,3\}$ with odd $p$ and carry out the process described previously, thus obtaining the tessellation $\{p,4,3,4\}$. As $p$ is odd, considering $\{p,4,3,4\}$ being the resulting hypergraph, it does not satisfy the restriction $(H_5)$, therefore, it does not satisfy the commutation rule, and, thus, it does not support the construction of subsystem codes. To fix this, we place a triangle inside each $p$-gon. These triangles cannot be placed randomly, but we will explain how to place them later. Finally, we color the triangles blue and alternate the other edges between red and green.

We denote the resulting hypergraph by $\Gamma_h$ and use it to construct the topological subsystem code. To prove that our construction of $\Gamma_h$ can be used to construct subsystem codes using the construction of \cite{terhal}, we will prove that $\Gamma_h$ satisfies the constraints $(H_1),\ldots,( H_5)$.

\begin{Theo}\label{ultimoteo}
The hypergraph $\Gamma_h$ obtained from the tessellation $\{p,4,3,4\}$ with $p\geq 7$ odd satisfies the constraints $(H_1),\ldots,(H_5)$ and, therefore, it gives rise to subsystem codes whose group formed by loop operators is given by
\begin{equation}
\G_{loop} = \{W(M);\ M\subseteq E\ \ \textrm{é um hiperciclo fechado}\},
\end{equation}
where $E$ is the set given by the edges of rank-$2$ and rank-$3$, and the gauge group $\G$ is given by $\G = C(\G_{loop})$.
\end{Theo}

\begin{proof}
Considering the triangles of tessellation $\{p,4,3,4\}$ as being rank-$3$ edges, we will only have rank-$2$ and rank-$3$ edges. Likewise, when we create the triangle inside the $p$-gon from points taken on three of its edges, we are producing one edge of rank-$3$, and three of rank-$2$. Therefore, $\Gamma_h$ satisfies $(H_1)$.

When we consider triangles to be rank-$3$ edges, it follows that every vertex of $\{p,4,3,4\}$ will be trivalent. Likewise, the three new vertices will be trivalent, as they divide an edge into two and have the new triangle as the third edge. Thus, $\Gamma_h$ satisfies $(H_2)$.

It follows from the construction of the tessellation $\{p,4,3,4\}$ that the constraint $(H_3)$ is satisfied. Even introducing the triangle inside each $p$-gon does not change this condition. 
 
The constraint $(H_4)$ is also satisfied since the rank-$3$ edges of $\{p,4,3,4\}$ do not intersect, and when we create the triangle inside the $p$-gon, we take the points on the edges so that there is no intersection with the other triangles. Therefore, the rank-$3$ edges of $\Gamma_h$ are two by two disjoint.

Finally, with the introduction of the triangle inside each $p$-gon, the restriction $(H_5)$ is satisfied. We just color all the triangles blue and alternate each $p+3$ edge between red and green. Therefore, $\Gamma_h$ gives rise to subsystem codes as constructed in \cite{terhal}.
\end{proof}

The operators $\overline{K}_{e'}$ will be given in the same way as in the previous cases. When we consider the ordinary graph $\overline{\Gamma}_h$ given by $\Gamma_h$, we will have for $e' =(u,v)\in\overline{E}$ that $\overline{K}_{e'} = Z_uZ_v$ if $e'\in\overline{E}_B$, $\overline{K} _{e'} = X_uX_v$ if $e'\in\overline{E}_R$, and $\overline{K}_{e'} = Y_uY_v$ if $e'\in\overline{E}_G$, where $\overline{E}_B$, $\overline{E}_R$ and $\overline{E}_G$ are the sets with all blue, red and green edges, respectively, of the graph $\overline{\Gamma}_h$.

In the case of tessellation $\{p,4,3,4\}$ with even $p$, there are three stabilizers for each $p$-gon, where only two are independent. One is formed by the border of the $p$-gon, and the other by the triangles involving the $p$-gon, and the edges necessary to close the hypercycle. As in our case $p$ is odd, when we introduce the triangle inside the $p$-gon, we create three more edges, making the $p$-gon, now with three more edges, a stabilizer satisfying the Theorem \ref{teosindrome}. However, when we introduce the triangle, this will increase a vertex in the hypercycles involving the triangles of the $p$-gons ``adjacent'' to the one where the triangle was introduced. For these hypercycles to provide stabilizers, we require that these triangles be fixed within the $p$-gons in such a way that an even number of vertices increases for each hypercycle $f_{\sigma_2}$, or no new vertices increase. In this way, we obtain three stabilizer generators (two independent) per $p$-gon.

\begin{lema} 
Each $p$-gon of our construction on the tessellation $\{p,4,3,4\}$ provides three hypercycles (two independent) and, consequently, three stabilizer generators (two independent).
\end{lema}

\begin{proof}
The proof of this result follows in a similar way to that seen in Section \ref{construcaomeussubcode}. The first hypercycle $f_{\sigma_1}$ is given by the $p+3$ edges of rank-$2$ created in $p$-gon, and the loop operator $W(f_{\sigma_1})$ is given by (\ref{hiperciclo1}). The second hypercycle $f_{\sigma_2}$ is given by the rank-$3$ edges around the $p$-gon, together with the rank-$3$ edge created inside the $p$-gon and the edges necessary to close the hypercycle. If no new edges have been increased in this hypercycle, the loop operator $W(f_{\sigma_2})$ will be given as (\ref{hiperciclo2}). If an even number of new edges have increased in this hypercycle, we will not be able to write the loop operator as in (\ref{hiperciclo2}), separating by colors, but it will continue as we saw in the case $\{2p_1,2p_2,4\}$, that $W(f_{\sigma_2})$ will be a stabilizer. The third and final hypercycle is given by the modulo $2$ sum of $f_{\sigma_1}$ and $f_{\sigma_2}$. The stabilizer will be given by the product of $W(f_{\sigma_1})$ by $W(f_{\sigma_2})$.
\end{proof}

Then, it follows that these loop operators satisfy Theorem \ref{teosindrome}. Therefore, they satisfy the condition $(C_2)$ of the topological subsystem code definition.

Thus, we will have three stabilizer generators, where two are independent for each $p$-gon. In the case of $p$ even, we had two independence relations, which are given in (\ref{bombinrelacao}). In the case of odd $p$, as we are not starting from a $3$-colorable tessellation, and by introducing the triangle inside each $p$-gon, there will be no independence relationship. Therefore, the quantity $s$ of independent stabilizing generators will be twice the quantity of $p$-gons.

To determine the value of $s$ and the other code parameters, we need the number of faces, vertices, and edges of the tessellation $\{p,3\}$, which are given in (\ref{fevcolors}), or it is,
\begin{eqnarray*}
n_f=\frac{12(g-1)}{p-6},\ \ n_e=\frac{6p(g-1)}{p-6}\ \ \textrm{e}\ \ n_v = \frac{4p(g-1)}{p-6}.
\end{eqnarray*}

Consider $F_{\Pol}$, $F_{\Tri}$, and $F_{\Q}$ the amount of $p$-gons, triangles, and squares of $\{p,4,3,4\}$, respectively. It follows from the construction process of $\{p,4,3,4\}$ that the number of $p$-gons is equal to the number of tessellation faces of $\{p,3\}$, the number of triangles is equal to the number of tessellation vertices of $\{p,3\}$, and the number of squares is equal to the number of tessellation edges of $\{p,3\}$, that is,
\begin{eqnarray*}
F_{\Pol}=\frac{12(g-1)}{p-6},\ \ F_{\Tri}=\frac{4p(g-1)}{p-6}\ \ \textbf{e}\ \ F_{\Q}=\frac{6p(g-1)}{p-6}.
\end{eqnarray*}

So, we have
\begin{equation}
s=2F_{\Pol}=\frac{24(g-1)}{p-6}.
\end{equation}

Introducing the triangle inside each $p$-gon, we obtain 3 more vertices for each $p$-gon. Thus, the quantity $n$ of physical qubits will be
\begin{equation}
n=(p+3)F_{\Pol}.
\end{equation}

To determine the quantity of gauge qubits, $r$, we will proceed as in the previous cases. To do this, we need the quantities of rank-$2$ and rank-$3$ edges of $\Gamma_h$, which are given by $E_2=(p+3)F_{\Pol}$ and $E_3=F_{\Tri} +F_{\Pol}$, respectively. Therefore, $\dim\G=E_2+2E_3=(p+3)F_{\Pol}+ 2F_{\Tri}+2F_{\Pol}$. As $s=2F_{\Pol}$ and $\dim \G=2r+s$, it follows that
\begin{eqnarray*}
2r+2F_{\Pol} & = & (p+3)F_{\Pol} + 2F_{\Tri}+2F_{\Pol} \Longleftrightarrow \\
    2r  & = & (p+3)F_{\Pol} + 2F_{\Tri}\\
        & = & 5F_{\Tri}+3F_{\Pol},       
\end{eqnarray*}
since $pF_{\Pol}=3F_{\Tri}$. Thus,
\begin{equation}
r = \frac{5F_{\Tri}+3F_{\Pol}}{2},
\end{equation}
where $r$ is a positive integer.

As we already know the values of $n$, $r$, and $s$, it is easy to determine the value of $k$, as it is given by $k = n - r - s$. So, we have to
\begin{eqnarray}
k & = & (p+3)F_{\Pol}-\frac{5F_{\Tri}+3F_{\Pol}}{2}-2F_{\Pol}\nonumber \\
  & = & \frac{F_{\Tri}-F_{\Pol}}{2},
\end{eqnarray}
where $k$ is a positive integer.

Finally, the distance $d$ follows as in the other cases, that is, it will be upper bound by the smallest number of triangles in a homologically non-trivial closed hypercycle, that is, it will be the smallest weight among the operators of $C(\G)\setminus\textbf{S}$.

Lemmas \ref{lemasub1} and \ref{lemasub2} are also valid here, and their proofs are done in a similar way. From this, it follows that this code satisfies the conditions $(C_1)$ and $(C_3)$. We have already seen that this code satisfies $(C_2)$, and it was shown in Theorem \ref{ultimoteo} that the graph $\Gamma_h$ satisfies $(H_1),\ldots,(H_5)$. Therefore, this proves the following theorem.

\begin{Theo}
Consider a tessellation $\{p,3\}$ with odd $p$ embedded in a compact orientable surface $\Msup$ of genus $g\geq 2$. By carrying out the process to obtain the tessellation $\{p,4,3,4\}$ and carrying out the construction described previously, we obtain a topological subsystem code with parameters
\begin{equation}
[[(p+3)F_{\Pol},\frac{F_{\Tri}-F_{\Pol}}{2}, \frac{5F_{\Tri}+3F_{\Pol}}{2},d\leq l]],
\end{equation}
where $F_{\Pol}$ and $F_{\Tri}$ are the amount of $p$-gons and triangles of the tessellation $\{p,4,3,4\}$, respectively, and $l$ is the number of triangles in a homologously non-trivial hypercycle.
\end{Theo}

To conclude, we present in Table \ref{t8subsistem} the parameters $n$, $k$, and $r$, obtained from the hypergraph $\Gamma_h$ originating from the tessellations $\{p,4,3 ,4\}$ on a compact orientable surface $\Msup$ with genus $g=2$ to $g=7$. We also present in each case the quantity $s$ of independent stabilizing generators together with the tessellations $\{p,4,3,4\}$ used to construct $\Gamma_h$.

\begin{table}[h!]
\centering
\begin{multicols}{2}
\begin{tabular}{|w{c}{2mm}|w{c}{18mm}|w{c}{3mm}|w{c}{28mm}|}
\hline	$g$	&	Tessellation	    &	$s$	&	$[[n,k,r,d]]$\\
\hline 2 & $\{7,4,3,4\}$ & 24 & $[[120,8,88,d]]$\\
\hline 2 & $\{9,4,3,4\}$ &	8 &	$[[48,4,36,d]]$\\
\hline 3 & $\{7,4,3,4\}$ & 48 &	$[[240,16,176,d]]$\\
\hline 3 & $\{9,4,3,4\}$ & 16 &	$[[96,8,72,d]]$\\
\hline 4 & $\{7,4,3,4\}$ & 72 &	$[[360,24,264,d]]$\\
\hline 4 & $\{9,4,3,4\}$ & 24 &	$[[144,12,108,d]]$\\
\hline 4 & $\{15,4,3,4\}$&	8 &	$[[72,8,56,d]]$\\
\hline 5 & $\{7,4,3,4\}$ & 96 &	$[[480,32,352,d]]$\\
\hline
\end{tabular}
\begin{tabular}{|w{c}{2mm}|w{c}{18mm}|w{c}{3mm}|w{c}{28mm}|}
\hline	$g$	&	Tesselação	    &	$s$	&	$[[n,k,r,d]]$\\
\hline 5 & $\{9,4,3,4\}$ & 32 &	$[[192,16,144,d]]$\\
\hline 6 & $\{7,4,3,4\}$ &120 &	$[[600,40,440,d]]$\\
\hline 6 & $\{9,4,3,4\}$ & 40 &	$[[240,20,180,d]]$\\
\hline 6 & $\{11,4,3,4\}$& 24 &	$[[168,16,128,d]]$\\
\hline 6 & $\{21,4,3,4\}$&	8 &	$[[96,12,76,d]]$\\
\hline 7 & $\{7,4,3,4\}$ & 144&	$[[720,48,528,d]]$\\
\hline 7 & $\{9,4,3,4\}$ & 48 &	$[[288,24,216,d]]$\\
\hline 7 & $\{15,4,3,4\}$& 16 &	$[[144,16,112,d]]$\\
\hline
\end{tabular}
\end{multicols}
\caption{Parameters of topological subsystem codes coming from tessellation $\{p,4,3,4\}$ with $p\geq 7$ odd for $g=2$ to $g=7$.}\label{t8subsistem}
\end{table}

\addcontentsline{toc}{chapter}{Referências}

\end{document}